\documentclass[11pt]{article}      
\usepackage[margin=1in]{geometry}  

\usepackage[dvipsnames]{xcolor}
\usepackage{mathrsfs}

\usepackage{amsmath}
\usepackage{amssymb}
\usepackage{amsthm}
\usepackage{thm-restate}
\usepackage{mdframed}
\usepackage{mathtools}
\usepackage[hidelinks]{hyperref}
\usepackage[capitalize]{cleveref}

\usepackage{float}
\usepackage{wrapfig}

\usepackage{cancel}
\usepackage{xspace}
\usepackage[textsize=tiny]{todonotes}
\usepackage[all,defaultlines=3]{nowidow}

\usepackage{microtype}

\newcommand{\N}{\mathbb{N}}
\mdfdefinestyle{myenvs}{%
  hidealllines=true,%
  nobreak=true, 
  leftmargin=0pt,
  rightmargin=0pt,
  innerleftmargin=0pt,
  innerrightmargin=0pt,
}
\newmdtheoremenv[style=myenvs]{lemma}{Lemma}[section]
\newmdtheoremenv[style=myenvs]{corollary}[lemma]{Corollary}
\newmdtheoremenv[style=myenvs]{theorem}[lemma]{Theorem}
\newmdtheoremenv[style=myenvs]{fact}[lemma]{Fact}
\newmdtheoremenv[style=myenvs]{proposition}[lemma]{Proposition}
\newmdtheoremenv[style=myenvs]{conjecture}[lemma]{Conjecture}
\newmdtheoremenv[style=myenvs]{definition}[lemma]{Definition}
\newmdtheoremenv[style=myenvs]{claim}[lemma]{Claim}
\newenvironment{claimproof}[1][\proofname]{%
  \begin{proof}[#1]%
}{%
  \end{proof}%
}

\newcommand{\Cc}[0]{\ensuremath{\mathscr{C}}\xspace}
\newcommand{\Dd}[0]{\ensuremath{\mathscr{D}}\xspace}

\newcommand{\CC}[0]{\mathrm{\mathscr{C}}}
\newcommand{\DD}[0]{\mathrm{\mathscr{D}}}
\newcommand{\NN}[0]{\mathrm{\mathbb{N}}}

\newcommand{\struc}[1]{\mathbb{#1}}
\newcommand{\flip}[1]{\mathsf{#1}}

\newcommand{\KK}{\mathcal{K}}

\renewcommand{\phi}{\varphi}
\renewcommand{\le}{\leqslant}
\renewcommand{\ge}{\geqslant}
\renewcommand{\leq}{\le}
\renewcommand{\geq}{\ge}

\newcommand{\Oof}{\mathcal{O}}

\newcommand{\sless}{{\scriptscriptstyle <}}
\newcommand{\sgreater}{{\scriptscriptstyle >}}

\newcommand{\funding}{
N.M. was supported by the German Research Foundation (DFG) with grant greement No. 444419611.
S.T. was supported by the project BOBR that is funded from the European Research Council (ERC) under the European Union’s Horizon 2020 research and innovation programme with grant agreement No. 948057.
}

\title{Indiscernibles and Flatness in Monadically Stable and \\Monadically NIP Classes\thanks{\funding
}} 

\author{
  Jan Dreier
  \\
  \small{TU Wien}
  \\
  \small{\texttt{dreier@ac.tuwien.ac.at}}
  \and
  Nikolas M\"ahlmann
  \\
  \small{University of Bremen}
  \\
  \small{\texttt{maehlmann@uni-bremen.de}}
  \and
  Sebastian Siebertz
  \\
  \small{University of Bremen}
  \\
  \small{\texttt{siebertz@uni-bremen.de}}
  \and
  Szymon Toru\'nczyk
  \\
  \small{University of Warsaw}
  \\
  \small{\texttt{szymtor@mimuw.edu.pl}}
}
\date{}

\begin{document}

\maketitle

\begin{abstract}
Monadically stable and monadically NIP classes of 
structures were initially studied in the context of model theory
and defined in logical terms.
They have recently attracted attention in the area of structural graph theory,
as they generalize notions such as nowhere denseness, bounded cliquewidth, and bounded twinwidth.

Our main result is the -- to the best of our knowledge first --  purely combinatorial characterization of monadically stable classes of graphs, in terms of a property 
dubbed \emph{flip-flatness}. A class $\CC$ of graphs is flip-flat if for every fixed radius $r$, every 
sufficiently large set of vertices of a graph  $G\in\CC$ contains a large subset of vertices  with mutual distance larger than $r$, where the distance is measured in some graph $G'$ that can be obtained from $G$ by performing a bounded number of flips that swap edges and non-edges within a subset of vertices.
Flip-flatness generalizes the notion of uniform quasi-wideness, which characterizes nowhere dense classes and had a key impact on the combinatorial and algorithmic treatment of nowhere dense classes. 
To obtain this result, we develop tools that also apply to the more general monadically NIP classes,
based on the notion of indiscernible sequences from model theory.
We show that in monadically stable and monadically NIP classes 
indiscernible sequences impose a strong combinatorial structure on their definable neighborhoods.
All our proofs are constructive and yield efficient algorithms.
\end{abstract}

\paragraph*{Acknowledgements.}
We thank \'Edouard Bonnet, Jakub Gajarsk\'y, Stephan Kreutzer, Amer E. Mouawad and Alexandre Vigny for their valuable contributions to this paper. In particular, we thank Jakub Gajarsk\'y and Stephan Kreutzer for suggesting the notion of flip-flatness and providing a proof for the cases $r=1,2$.

\newpage


\section{Introduction}
An important open problem in structural and algorithmic graph theory
is to characterize those hereditary graph classes for which  the model checking 
problem for first-order logic is tractable\footnote{more precisely, \emph{fixed parameter-tractable}, that is, solvable in time $f(|\phi|)\cdot |G|^c$,
where $\phi$ is the input formula and $G$ is the input graph, 
for some function $f:\mathbb N\to\mathbb N$ and constant $c$}~\cite[Section 8.2]{grohe2008logic}.
A result of Grohe, Kreutzer, and Siebertz~\cite{grohe2017deciding} 
states that for monotone graph classes 
(that is, classes closed under removing vertices and edges),
 the limit of tractability is precisely captured by the notion of nowhere denseness,
introduced by Ne\v set\v ril and Ossona de Mendez~\cite{nevsetvril2011nowhere}. 
Examples of nowhere dense classes include the class of planar graphs,
all classes that exclude a fixed minor, and classes with bounded expansion.
Whereas these classes are sparse (for instance, they exclude some fixed biclique as a subgraph),
 the aforementioned problem seeks to generalize the result of 
Grohe, Kreutzer, and Siebertz to classes that are not necessarily sparse.
Indeed, there are known hereditary graph classes
that are not sparse, and 
for which 
the model checking problem is tractable, such as transductions of classes of bounded local cliquewidth~\cite{bonnet2022model}, transductions of nowhere dense classes~\cite{ssmc}, or classes of ordered graphs (that is, graphs equipped with a total order)
of bounded twinwidth~\cite{bonnet2022twin}.

So far, a complete picture is understood in two contexts:
for monotone graph classes, where tractability coincides with nowhere denseness,
and for hereditary classes of ordered graphs, where tractability coincides with 
bounded twinwidth. Despite the apparent dissimilarity of the combinatorial definitions
of nowhere denseness and bounded twinwidth, those notions can be 
alternatively characterized in a uniform way in logical terms
by the following notion, originating in model theory.
Unless mentioned otherwise, all formulas are first-order formulas.

Say that a class $\CC$ of graphs \emph{transduces} a class $\DD$ of graphs
if for every $H\in \DD$ there is some $G\in\CC$ 
from which $H$ can be obtained by performing the following steps:
(1) coloring the vertices of $G$ arbitrarily
(2) interpreting a fixed formula $\phi(x,y)$ (involving 
the edge relation and unary relations for the colors), thus yielding a new graph \(\phi(G)\)
with the same vertices as $G$ and edges $uv$ such that $\phi(u,v)$ holds, and finally
(3) taking an induced subgraph of \(\phi(G)\).
The transducability relation on graph classes is transitive, and classes 
that \emph{do not} transduce the class of all graphs are called \emph{monadically~NIP}.
For instance, the class of all bipartite graphs transduces the class of all graphs:
to obtain an arbitrary graph~$G$, consider its $1$-subdivision, obtained 
by placing one vertex on each edge of $G$, thus yielding a bipartite graph $H$;
then the formula $\phi(x,y)$ expressing that $x$ and $y$ have a common neighbor
defines a graph on $V(H)$ containing $G$ as an induced subgraph. Hence,
the class of bipartite graphs is not monadically NIP.
On the other hand, all the graph classes mentioned earlier --
nowhere dense classes and transductions thereof, classes of bounded twinwidth,
or transductions of classes with bounded local cliquewidth -- 
are monadically NIP. This suggests that monadic NIP 
might constitute the limit of tractability of the model checking problem.
More precisely, the following has been conjectured\footnote{To the best of our knowledge the conjecture was first explicitly discussed during the open problem session of the  Algorithms, Logic and Structure Workshop in Warwick, in 2016, see \cite{warwick-problems}.}.
\begin{conjecture}[\cite{warwick-problems}]\label{conj:mNIP}
Let  $\CC$ be a hereditary class of graphs. Then the model checking problem for first-order logic 
is fixed parameter tractable on $\CC$ if and only if $\CC$ is monadically NIP.
\end{conjecture}

 Quite remarkably, among monotone graph classes, 
 classes that are monadically NIP correspond precisely to nowhere dense classes~\cite{adler2014interpreting},
 and among hereditary graph classes of ordered graphs,
 classes that are monadically NIP correspond precisely to classes of bounded twinwidth~\cite{bonnet2022twin}.
One may tweak the definition of monadic NIP classes
by considering other logics than first-order logic.
For instance, for the counting extension $\mathsf{CMSO_2}$ of monadic second-order logic,
one recovers precisely the notion of classes of bounded cliquewidth \cite{COURCELLE200791},
or classes of bounded treewidth if only monotone classes are considered.
Thus, variations on the definition of monadic NIP recover 
important notions from graph theory: nowhere denseness, bounded twinwidth, bounded treewidth, and bounded 
cliquewidth.

 Note that both implications in Conjecture~\ref{conj:mNIP} remain open.
 The conjecture is so far confirmed 
for monotone graph classes~\cite{grohe2017deciding} (where monadically NIP classes are exactly 
the nowhere dense classes) and for hereditary classes of ordered graphs~\cite{bonnet2022twin}, tournaments~\cite{twinwidthTournaments}, interval graphs and permutation graphs \cite{twwVIII}, 
(where monadically NIP classes are exactly the classes of bounded twinwidth).
As a special important case, the conjecture predicts that all \emph{monadically stable} 
graph classes are tractable.
A class $\CC$ is {monadically stable} if it does not transduce 
the class of all half-graphs, that is, graphs with vertices $a_1,b_1,\ldots,a_n,b_n$
such that $a_i$ is adjacent to $b_j$ if and only if $i\le j$.
Although much more restrictive than monadically NIP classes,
monadically stable classes include all nowhere dense classes~\cite{adler2014interpreting},
and hence also all classes $\DD$ that transduce in a nowhere dense class $\CC$ (called structurally nowhere dense classes~\cite{gajarsky2020first}).
Those include dense graph classes, such as for instance squares of planar 
graphs.
In fact, it is conjectured~\cite{nevsetvril2021rankwidth} that every monadically stable class of graphs
is structurally nowhere dense.

These outlined 
connections between structural graph theory and model theory have recently triggered the interest to
generalize combinatorial and algorithmic results from nowhere dense
classes to structurally nowhere dense, monadically stable
and monadically NIP classes, and ultimately to efficiently solve the
model checking problem for first-order logic on these
classes~\cite{bonnet2022model,bonnet2022twin,gajarsky2020first,dreier2022treelike,gajarsky2021stable,kwon2020low,
  nevsetvril2021rankwidth,nevsetvril2022structural, nevsetvril2020linear}.
  Logical results on monadically stable and monadically NIP classes in  model theory
include~\cite{baldwin1985second,shelah1986monadic,braunfeld2021characterizations,blumensath2011}.  

\subsection*{Contribution}

As discussed above, many central graph classes such as those with bounded cliquewidth, twinwidth or nowhere dense classes can be characterized both from a structural (i.e., graph theoretic) and a logical perspective.
While monadically stable and monadically NIP graph classes are naturally defined via logic, structural characterizations have so far been elusive.
In this work we take a step towards a structure theory for monadically stable and monadically NIP classes of graphs, which is the basis for their future algorithmic and combinatorial treatment, in particular, a tool for approaching \cref{conj:mNIP}, as we discuss later.

\begin{figure}[h]
  \begin{center}
  \includegraphics[scale=0.8]{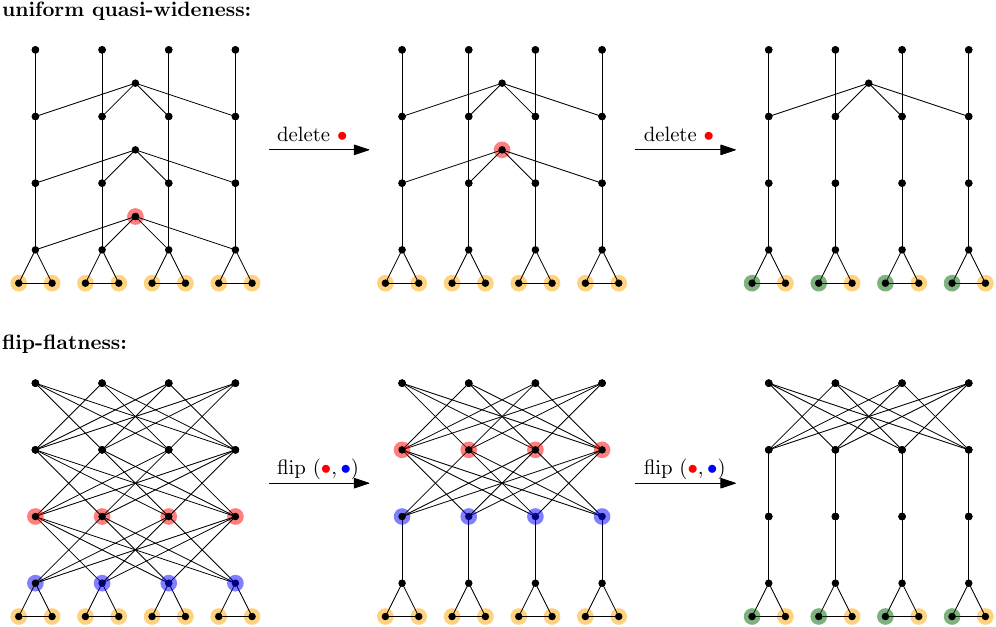}
  \end{center}
  \caption{
    An example of uniform quasi-wideness (flip-flatness).
    Among the \textcolor{Dandelion}{yellow} vertices, we find a still large set of \textcolor{ForestGreen}{green} vertices, that is distance-$7$ independent after deleting a bounded number of \textcolor{red}{red} vertices (performing a bounded number of flips between sets of \textcolor{red}{red} and \textcolor{blue}{blue} vertices).
    Two key properties are illustrated: 1. the higher the desired independence distance, the more operations have to be performed; 2. we cannot hope to make all the \textcolor{Dandelion}{yellow} vertices distance-$r$ independent with a bounded number of operations. 
    }
  \label{fig:wideness}
\end{figure}

\paragraph*{Flatness.}
Our main result, \cref{thm:fuqw}, provides a purely combinatorial characterization of monadically stable graph classes in terms of \emph{flip-flatness}.
Flip-flatness generalizes \emph{uniform quasi-wideness}\footnote{
Another reasonable name for flip-flatness would be \emph{flip-wideness}. 
We avoided this name to prevent confusion with the recently introduced graph parameter \emph{flip-width} \cite{torunczyk2023flipwidth}, which is studied in the same context.}, introduced by Dawar in~\cite{dawar2010homomorphism} in his study of homomorphism preservation properties.
Uniform quasi-wideness was proved by Ne\v{s}et\v{r}il and Ossona de Mendez~\cite{nevsetvril2011nowhere} to characterize nowhere dense graph classes
and is a key tool in the combinatorial and algorithmic treatment of these classes. 
To foster the further discussion, let us formally define this notion. 

\pagebreak
A set of vertices is \emph{distance-\(r\) independent} in a graph if 
any two vertices in the set are at 
distance greater than \(r\).
A class of graphs $\CC$ is \emph{uniformly quasi-wide}
if for every $r \in \NN$ there exists a function $N_r : \NN \rightarrow \NN$ and a constant $s\in \NN$ with the following property.
For all $m\in\NN$, $G\in \CC$, and $A\subseteq V(G)$ with $|A|\geq N_r(m)$, there exists $S\subseteq V(G)$ with $|S|\leq s$ and $B\subseteq A \setminus S$ with $|B|\geq m$ such that $B$~is distance-$r$ independent in $G-S$ (see the top of \Cref{fig:wideness}). Intuitively, for uniformly quasi-wide classes, in sufficiently large sets one can find large subsets that are distance-$r$ independent after the deletion of a constant number of vertices. Uniform quasi-wideness is only suitable for the treatment of sparse graphs. Already very simple dense graph classes, such as the class of all cliques, are not uniformly quasi-wide.

Inspired by the notion of uniform quasi-wideness, Jakub Gajarsk\'y and Stephan Kreutzer proposed the new notion of \emph{flip-flatness}. 
Roughly speaking, flip-flatness generalizes uniform quasi-wideness by replacing in its definition the deletion of a bounded size set of vertices by the inversion of the edge relation between a bounded number of (arbitrarily large) vertex sets.
Let us make this definition more precise. 
A \emph{flip} in a graph $G$ is specified by a pair of sets $\mathsf{F} = (A,B)$ with $A,B \subseteq V(G)$. We write $G \oplus \mathsf{F}$ for the graph 
with the same vertices as $G$, and edges $uv$ such that $uv \in E(G)$ xor $(u,v) \in (A \times B) \cup (B \times A)$. For a set $F =\{\flip F_1,\ldots,\flip F_n\}$ of flips, we write $G\oplus F$ for the graph $G \oplus \flip F_1 \oplus \dots \oplus \flip F_n$. 

\begin{restatable}{definition}{defflipwide}\label{def:flipwide}
  A class of graphs $\CC$ is \emph{flip-flat}
  if for every $r \in \NN$ there exists a function $N_r : \NN \rightarrow \NN$ and a constant $s_r\in \NN$ such that for all $m\in\NN$, $G\in \CC$, and $A\subseteq V(G)$ with $|A|\geq N_r(m)$, there exists a set $F$ of at most $s_r$ flips and $B\subseteq A$ with $|B|\geq m$ such that $B$ is distance-$r$ independent in~$G \oplus F$.
\end{restatable}
Intuitively, for flip-flat classes in sufficiently large sets one can find large subsets that are distance-$r$ independent after a constant number of flips (see bottom of \cref{fig:wideness}).
For example the class of all cliques is flip-flat, requiring only a single flip to make the whole vertex set distance-$\infty$ independent.
Our main result is the following purely combinatorial characterization of monadic stability.

\begin{restatable}{theorem}{thmfuqw}\label{thm:fuqw}
  A class of graphs is monadically stable if and only if it is flip-flat. 
\end{restatable}

Notably, our proof is algorithmic and yields polynomial bounds in the following sense. 

\begin{restatable}{theorem}{thmfuqwforward}\label{thm:fuqw_forward}
  Every monadically stable class $\CC$ of graphs is flip-flat, where for every~$r$, the function $N_r$ is polynomial. 
  Moreover, given an $n$-vertex graph $G\in\CC$, $A\subseteq V(G)$, and $r\in\NN$,
  we can compute a subset $B \subseteq A$ and a set of flips~$F$ that makes $B$ distance-$r$ independent in $G\oplus F$ in time $\Oof(f_\CC(r) \cdot n^3)$ for some function~$f_\CC$. 
\end{restatable}



Just as uniform quasi-wideness provided a key tool for the algorithmic treatment of nowhere dense graph classes, in particular for first-order model checking~\cite{grohe2017deciding}, we believe that the characterization by flip-flatness will provide an important step towards the algorithmic treatment of monadically stable classes. 
We leave as an open question whether a similar characterization of monadically NIP classes exists.

\paragraph*{Indiscernible sequences.}
Our study is based on \emph{indiscernible
  sequences}, which are a fundamental tool in model
theory. 
An indiscernible sequence is
a (finite or infinite) sequence $ I=(\bar a_1,\bar a_2,\ldots)$ of tuples of equal
length of elements of a fixed (finite or infinite) structure, such that 
any two finite subsequences of $I$
that have equal length, satisfy the same
formulas.
More generally, for a
set of formulas $\Delta$, the 
sequence $I$ is \emph{$\Delta$-indiscernible} if for each
formula $\phi(\bar x_1,\ldots, \bar x_k)$ from $\Delta$ either all subsequences of $I$
of length~$k$ satisfy~$\phi$, or no subsequence of length $k$
satisfies~$\phi$.
For example,
any enumeration of vertices forming a clique or an independent set in a graph is $\Delta$-indiscernible for $\Delta = \{E(x_1,x_2)\}$ containing only the edge relation.
Also, any increasing sequence of rationals in the structure $(\mathbb{Q},<)$ is $\Delta$-indiscernible 
for $\Delta$ being the set of all formulas.
We use indiscernible sequences to obtain new insights about monadically stable and monadically NIP classes.

It was shown by Blumensath~\cite{blumensath2011}
that in monadically NIP classes, 
any fixed element interacts with the tuples of an indiscernible sequence in a very regular way.
We give a new finitary proof of Blumensath's result. Building on this, we develop our main technical tool, \cref{thm:disjoint_families_nip}, where 
we show that the regular properties of the indiscernible sequences extend to their disjoint definable neighborhoods (see \cref{sec:families} for details).
A result similar to \Cref{thm:disjoint_families_nip}, also for disjoint definable neighborhoods in monadically NIP classes, plays a key role in~\cite{bonnet2022twin,OGBT}.

Apart from powering our algorithmic proof of flip-flatness, \cref{thm:disjoint_families_nip} has already found further applications in monadically stable classes of graphs: 
it was recently used to obtain improved bounds for Ramsey numbers \cite{mstable-ramsey} and to prove an algorithmic game-characterization of these classes, called \emph{Flipper game} \cite{flipper-game}.
The paper \cite{flipper-game} provides two proofs for this game characterization. 
One is constructive and algorithmic and builds on our work.
The other builds on model theoretic tools;
it is non-constructive but self-contained and highlights additional properties of monadically stable graph classes, including a second (though non-constructive) proof of \cref{thm:fuqw}.
The algorithmic version of the Flipper game plays a crucial role for proving that the first-order model checking problem is fixed parameter tractable for structurally nowhere dense classes of graphs~\cite{ssmc}. 
This suggests that our developed techniques may be an important tool towards resolving \cref{conj:mNIP}.

\section{Preliminaries}\label{sec:prelims}

We use standard notation from graph theory and model theory
and refer to~\cite{Diestel} and~\cite{Hodges} for extensive background. 
We write $[m]$ for the set of integers $\{1,\ldots,m\}$.
For a function $f : \NN \rightarrow \NN$ we write~$f^k(n)$ for the $k$ times iterated application of the function $f$ to the input~$n$.

\paragraph*{Relational structures and graphs.}
A \emph{(relational) signature} $\Sigma$ is a set of relation symbols, each 
with an associated non-negative integer, called its \emph{arity}.
A \emph{$\Sigma$-structure} $\struc A$ consists of a \emph{universe}, 
which is a non-empty, possibly infinite set, and \emph{interpretations} of the symbols from the
signature: each relation symbol of arity $k$ is interpreted as
a $k$-ary relation over the universe. 
By a slight abuse of notation, we do not differentiate between structures and their universes and between relation symbols and their interpretations. 
By $\CC$ we denote classes of structures over a fixed signature. Unless indicated otherwise,
$\CC$ may contain finite and infinite structures.

A \emph{monadic extension} $\Sigma^+$ of a signature $\Sigma$
is any extension of $\Sigma$ by unary predicates. A unary predicate will
also be called a \emph{color}. 
A \emph{monadic expansion} or \emph{coloring} of a \mbox{$\Sigma$-structure~$\struc A$} is a $\Sigma^+$-structure~$\struc A^+$, where $\Sigma^+$
is a monadic extension of $\Sigma$, such 
that $\struc A$ is the \emph{$\Sigma$-reduct} of $\struc A^+$, that is, 
the $\Sigma$-structure obtained from $\struc A^+$ by
removing all relations with symbols in~$\Sigma^+\setminus \Sigma$.
When $\Cc$ is a class of $\Sigma$-structures and $\Sigma^+$ is a 
monadic extension of $\Sigma$, we 
write~$\Cc[\Sigma^+]$ for the class
of all possible $\Sigma^+$-expansions of structures from $\Cc$. 

A \emph{graph} is a finite structure over the signature consisting of a binary
relation symbol~$E$, interpreted as the symmetric and irreflexive edge relation.

\paragraph*{First-order logic.}

We say that
two tuples $\bar a, \bar b$ of elements are \emph{$\phi$-connected}
in a structure $\struc A$ if $\struc A\models\phi(\bar a, \bar b)$. We call
the set $N_\phi^{\struc A}(\bar a)=\{\bar b \in \struc A^{|\bar y|}~:~\struc A\models\phi(\bar a, \bar b)\}$ the \emph{$\phi$-neighborhood} of $\bar a$. 
We simply write~$N_\phi(\bar a)$ when $\struc A$ is clear from the context. 

Let $\Phi(\bar x)$ be a finite set of formulas with free variables $\bar x$.
A \emph{$\Phi$-type} is 
a conjunction $\tau(\bar x)$ of formulas in $\Phi$
or their negations, such that every formula in $\Phi$
occurs in $\tau$ either positively or negatively.
More precisely, $\tau(\bar x)$ is a formula of the following form, for some subset $A\subseteq\Phi$:
\[\bigwedge_{\phi \in A} \phi(\bar x)\wedge  \bigwedge_{\phi \in \Phi\setminus A} \hspace{-2mm}\neg\phi(\bar x).
\]

Note that for every  $|\bar x|$-tuple $\bar a$ of elements of $\struc A$,
there is exactly one $\Phi$-type $\tau(\bar x)$ such that $\struc A\models \tau(\bar a)$.

\paragraph*{Stability and NIP.}
We refer to the textbooks \cite{baldwin2017fundamentals,pillay2008introduction,shelah1990classification, tent2012course}
for extensive background on classical stability theory. 
While we already defined monadic stability and NIP in terms of transductions in the introduction, let us also give the equivalent original definition.
A formula $\phi(\bar x,\bar y)$ has the \emph{$k$-order property}
on a class $\Cc$ of structures if there are $\struc A\in \Cc$ and 
two sequences $(\bar{a}_i)_{1\leq i\leq k}$, $(\bar{b}_i)_{1\leq i\leq k}$ 
of tuples of elements of $\struc A$, such that
for all $i,j\in[k]$
  \[\struc A \models \phi(\bar{a}_i,\bar{b}_j) \quad \Longleftrightarrow\quad  i \leq j.\]
  
The formula $\phi$ is said to have the \emph{order property} on $\Cc$ 
if it has the $k$-order property for all $k\in \mathbb{N}$.
The class $\Cc$
is called \emph{stable} if no formula has the order-property on~$\Cc$. 
A class $\Cc$ of $\Sigma$-structures is \emph{monadically stable}
if for every monadic extension $\Sigma^+$ of $\Sigma$ the class 
$\Cc[\Sigma^+]$ is stable. 

Similarly,
a formula $\phi(\bar x,\bar y)$ has the \emph{$k$-independence property}
on a class $\Cc$ if there are $\struc A \in \CC$, a size $k$ set $A \subseteq \struc A^{|\bar x|}$ and a sequence $(\bar b_J)_{J\subseteq A}$ of tuples of elements of 
$\struc A$ such that for all $J\subseteq A$ and for all $\bar a\in A$
\[\struc A \models \phi(\bar a, \bar b_J) \quad \Longleftrightarrow \quad 
\bar a\in J.\]

We then say that $A$ is \emph{shattered} by $\phi$. 
We define the \emph{independence property} (IP), classes with the \emph{non-independence property} (NIP classes), and \emph{monadically NIP} classes as expected.
Note that every (monadically) stable class
is (monadically) NIP. 
Baldwin and Shelah proved that in the definitions of monadic stability and monadic NIP, one can alternatively rely on formulas $\phi(x,y)$ with just a pair of singleton variables, instead of a pair of tuples of variables.

\begin{lemma}[{\cite[Lemma~8.1.3, Theorem~8.1.8]{baldwin1985second}}]\label{lem:mon-stab1}
  A class $\Cc$ of $\Sigma$-structures 
  is monadically stable (monadically NIP) if and only if for every monadic
  extension $\Sigma^+$ of $\Sigma$ every $\Sigma^+$-formula~$\phi(x,y)$ 
  is stable (NIP) over~$\Cc[\Sigma^+]$. 
\end{lemma}

We call a relation $R$ \emph{definable} on a structure $\struc A$  if
$R = \{\bar a \in \struc A^{|\bar x|} : \struc A\models \phi(\bar a)\}$
for some formula~$\phi(\bar x)$.
The following is immediate from the previous definitions.

\begin{lemma}\label{lem:definable-expansions}
Let $\Cc$ be a monadically stable (monadically NIP) class of 
$\Sigma$-structures, let~$\Sigma^+$ be a monadic extension of 
$\Sigma$ and let $\Dd$ be the hereditary closure of 
any expansion of $\Cc[\Sigma^+]$ by definable relations. Then also $\Dd$
is monadically stable (monadically NIP). 
\end{lemma}

A formula $\phi(x,\bar y)$ has \emph{pairing index} $k$ on a class $\Cc$ of
structures if there is $\struc A\in \Cc$ and two sequences 
$(a_{ij})_{1\leq i<j\leq k}$ and $(\bar b_i)_{1\leq i\leq k}$ such that for
all $1\leq i<j\leq k$ and $\ell\in [k]$
\[\struc A \models \phi(a_{ij}, \bar b_{\ell})\quad \Longleftrightarrow\quad \ell \in \{i,j\}.\]

We require that $x$ is a single free variable, while $\bar y$ is allowed
to be a tuple of variables. 
Intuitively, 
from a graph theoretic perspective, if a formula has unbounded pairing index on a class $\CC$, then it can encode
arbitrarily large $1$-subdivided cliques in $\CC$, where the principal vertices are represented by the tuples $\bar b_\ell$ and the subdivision vertices are represented by single elements $a_{ij}$.
As discussed in the introduction, this is sufficient to encode arbitrary graphs in~$\CC$ by using an additional color predicate. 
In this case, $\CC$ cannot be monadically NIP. 
The above reasoning is formalized in the following characterization. 

\begin{restatable}{lemma}{pairingIndex}\label{cor:pairing-index}
A class $\Cc$ of $\Sigma$-structures is monadically NIP if and only if 
for every monadic extension~$\Sigma^+$ of $\Sigma$ every 
$\Sigma^+$-formula~$\phi(x,\bar y)$ has bounded pairing index
on $\Cc[\Sigma^+]$.
\end{restatable}
\begin{proof}
If a formula $\phi(x,\bar y)$ has unbounded pairing index on $\CC[\Sigma^+]$,
then the formula $\psi(\bar x, \bar y) := \exists z. R(z) \land \phi(z,\bar x) \land \phi(z,\bar y)$
can shatter arbitrary large sets over the class of monadic expansions of $\CC[\Sigma^+]$ with one extra unary predicate $R$, making this class monadically independent.
By \Cref{lem:definable-expansions}, $\CC$ is then also monadically independent.

If a class $\CC$ is monadically independent,
then by \Cref{lem:mon-stab1}, some formula $\phi(x,y)$ can shatter arbitrary large sets over $\CC[\Sigma^+]$ for some monadic extension $\Sigma^+$.
For every $k$ we can find a set $\{b_1,\dots,b_k\} \subseteq \struc A \in \CC[\Sigma^+]$ shattered by $\phi(x,y)$ via elements $(a_J)_{J\subseteq [k]}$.
The sequences $(a_{\{b_i,b_j\}})_{1\leq i<j\leq k}$ and $(b_i)_{1\leq i\leq k}$
witness that $\phi(x,y)$ has pairing index at least $k$ on $\CC[\Sigma^+]$.
\end{proof}


\paragraph*{Indiscernible sequences.}
Let $\struc A$ be a $\Sigma$-structure and $\phi(\bar x_1,\ldots, \bar x_m)$ a formula. 
We say a sequence $(\bar a_i)_{1\leq i\leq n}$ of tuples from $\struc A$ (where all $\bar x_i$ and $\bar a_j$ have the same length) 
is a \mbox{\emph{$\phi$-indis\-cernible sequence}} of 
length $n$, if for every two sequences of indices
  $i_1 < \dots < i_m$ and
  $j_1 < \dots < j_m$ from $[n]$ we have 
  \[
  \struc A \models 
  \phi(\bar a_{i_1}, \ldots, \bar a_{i_m}) 
  \quad \Longleftrightarrow \quad
  \struc A \models \phi(\bar a_{j_1}, \ldots, \bar a_{j_m}).
  \]
For a set of formulas $\Delta$ we call a sequence
\emph{$\Delta$-indiscernible} if it is $\phi$-indiscernible for 
all $\phi\in \Delta$. 
For any sequence $I=(\bar a_i)_{1\leq i\leq n}$ in a structure \(\struc A\) we define the 
\emph{Ehrenfeucht-Mostowski type (EM-type)} of~$I$ as the
set of all formulas $\phi$ such that $\struc A\models \phi(\bar a_{i_1}, \ldots, \bar a_{i_m})$
for all $i_1<\ldots<i_m$. A sequence $I$ is $\Delta$-indiscernible if
and only if for all formulas $\phi\in \Delta$ either $\phi$ or $\neg\phi$ belongs to 
$\mathrm{EM}(I)$. 
Note that a subsequence $J$ of a $\phi$-indiscernible sequence $I$
is again a $\phi$-indiscernible sequence and that the
$\mathrm{EM}$-type of any subsequence $J$ of a sequence $I$ is a \emph{superset} of the $\mathrm{EM}$-type of $I$.

%
For finite $\Delta$, by (iteratively applying) Ramsey's theorem we can extract
a $\Delta$-indiscernible sequence from any sequence. In general
structures however, in order to extract a large $\Delta$-indiscernible
sequence we must initially start with an enormously large sequence.  To the best of our knowledge the following theorem goes back to Ehrenfeucht and Mostowski~\cite{ehrenfeucht1956models}. 

\begin{theorem}\label{thm:indiscernible_ramsey}
    Let $\Delta$ be a finite set of formulas.
    There exists a function $f$ such that every sequence of elements of length at least $f(m)$ (in a structure with a signature matching $\Delta$)
    contains a $\Delta$-indiscernible subsequence of length $m$.
\end{theorem}

In stable classes we can efficiently find polynomially large indiscernible sequences of elements (see also \cite[Theorem 3.5]{regularity_lemmas}). 
The theorem is stated for graphs, but it easily extends to general structures. 

\begin{theorem}[{\cite[Theorem 2.8]{kreutzer2018polynomial}}]\label{thm:extract_indiscernibles}
  Let $\CC$ be a stable class of graphs and let $\Delta$ be a finite
  set of formulas.  There is a polynomial $p(x)$ such that
  for all $G\in \CC$, every positive integer $m$ and every sequence
  $J$ of vertices of $G$ of length $\ell=p(m)$, there
  exists a subsequence $I$ of length~$m$ which is
  $\Delta$-indiscernible.

  Furthermore, there is an algorithm that given an $n$-vertex graph $G\in \CC$
  and a sequence $(v_i)_{1\leq i\leq \ell}$, computes a
  $\Delta$-indiscernible subsequence of length at
  least $m$. The running time of the algorithm is in
  $O(|\Delta| \cdot k \cdot \ell^{k+1} \cdot n^{q}\cdot a(n)\cdot \lambda(\Delta))$, where $k$
  is the maximal number of free variables, $q$ is the maximal quantifier-rank of
  a formula of $\Delta$, $a(n)$ is the time required to test adjacency between two vertices and
  $\lambda(\Delta)$ is the length of a longest formula in $\Delta$.    
\end{theorem}

By starting with a sequence $(v_i)_{1\leq i\leq \ell'}$ of length $\ell' = p(m)^{k^2+k}$, we know that the number of vertices~$n$ in the host graph is at least $\ell'$.
We may then truncate the sequence to length $\ell=p(m)$ and execute the algorithm of \Cref{thm:extract_indiscernibles}.
The factor $\ell^{k+1} = \ell'^{1/k}$ in the run time can therefore be bounded by $n^{1/k}$,
while still guaranteeing that the starting length $\ell' = p(m)^{k^2+k}$ is only polynomial in $m$.
Looking at the proof in \cite{kreutzer2018polynomial}, we see that
the time needed to evaluate $G \models \phi(\bar v)$ for formulas $\phi \in \Delta$ and tuples $\bar v$
is bounded by $\Oof(n^{q}\cdot a(n)\cdot \lambda(\Delta))$.
Over the course of this paper, we will construct indiscernibles with respect to formulas
that can be evaluated much faster.
To capture this, we only assume a generic bound $f_{\CC\Delta}(n)$ on the evaluation time for formulas in $\Delta$ on $n$-vertex graphs from~$\CC$.
Applying these two modifications to \Cref{thm:extract_indiscernibles}, we immediately get the following statement.

\begin{theorem}\label{thm:poly_seqs}
  Let $\CC$ be a stable class of graphs and let $\Delta$ be a finite set of formulas.
  Assume that we can evaluate, given $\phi \in \Delta$, an $n$-vertex graph 
  $G \in \CC$  and a tuple~$\bar v$ in $G$, whether $G \models \phi(\bar v)$ in time  
  $f_{\CC\Delta}(n)$.

  Then there is an integer $t$ and an algorithm that 
  finds in any sequence of at least $m^t$ vertices in an $n$-vertex graph $G \in\CC$ 
  a $\Delta$-indiscernible subsequence of length $m$.
  The running time of the algorithm is 
  $\Oof(|\Delta| \cdot k \cdot n^{1/k} \cdot f_{\CC\Delta}(n))$, where $k$
  is the maximal number of free variables of formulas in $\Delta$.
\end{theorem}

\section{Indiscernibles in monadically NIP and monadically stable classes}\label{apx:indiscernibles}

In classical model theory, instead of classes of finite structures, usually infinite structures are studied.
For a single infinite structure $\struc A$, we say 
$\struc{A}$ is (monadically) stable/NIP, if the class $\{\struc{A}\}$ is (monadically) stable/NIP.
In this context, an indiscernible sequence is usually assumed to be of infinite length and 
indiscernible over the set of all first-order formulas.
Using $\Omega$ to denote the set of all first-order formulas, we will denote the latter property as \emph{$\Omega$-indiscernibility}.

It is well-known that indiscernible sequences can be used to characterize stability and NIP for infinite structures.
To characterize NIP, define the \emph{alternation rank} of a formula~$\phi(\bar x,\bar y)$ over a sequence $I$ in a structure $\struc A$ as the maximum~$k$ such that there
exists a (possibly non-contiguous) subsequence $(\bar a_1,\ldots,\bar a_{k+1})$ of $I$ and a tuple $\bar b \in \struc A^{|\bar x|}$ such that for all $i \in [k]$ we have $\struc A \models \phi(\bar b,\bar a_{i}) \iff \struc A \models \neg\phi(\bar b,\bar a_{i+1})$.
A structure $\struc A$ is NIP if and only if every formula has finite alternation rank over every $\Omega$-indiscernible sequence in~$\struc A$ \mbox{\cite[Theorem 12.17]{Poizat}}.
Stable classes can be characterized in a similar way. 
We say the \emph{exception rank} of a formula~$\phi(\bar x,\bar y)$ over a sequence $I$ in a structure $\struc A$ is the maximum~$k$ such that
there exists a tuple $\bar b \in \struc A^{|\bar x|}$ 
such that for $k$ distinct tuples $\bar a_i \in I$ we have $\struc A\models\phi(\bar b, \bar a_i)$
and for $k$ other distinct tuples $\bar a_i \in I$ we have $\struc A\models\neg\phi(\bar b, \bar a_i)$.
A structure $\struc A$ is stable if and only if every formula has finite exception rank over every $\Omega$-indiscernible sequence in $\struc A$ \mbox{\cite[Corollary 12.24]{Poizat}}.

Hence, NIP and stablility can be characterized by the interaction of tuples of elements with indiscernible sequences. 
For \emph{monadically} NIP structures, Blumensath \cite{blumensath2011} shows that the interaction of single elements with indiscernible sequences is even more restricted.

\begin{restatable}[{\cite[{Corollary 4.13}]{blumensath2011}}]{theorem}{thmblumensath}\label{thm:blumensath}
  In every monadically NIP structure $\struc A$, for every formula~$\phi(x,\bar y)$, where $x$ is a single free variable, 
  and $\Omega$-indiscernible sequence $(\bar a_i)_{i \in \N}$ of $|\bar y|$-tuples the following holds:
  for every element $b\in \struc A$
  there exists an exceptional index $\mathrm{ex}(b) \in \N$ and 
  two truth values $t_<(b),t_>(b) \in \{0,1\}$ such that
    \begin{align*}
      &\text{for all } i < \mathrm{ex(b)}: 
      \struc A \models \phi(b,\bar a_i) \hspace{2mm} \iff \hspace{2mm} t_<(b) = 1, \text{ and } \\
      &\text{for all } i > \mathrm{ex(b)}: 
      \struc A \models \phi(b,\bar a_i)\hspace{2mm} \iff \hspace{2mm} t_>(b)=1.
    \end{align*}
\end{restatable}

See \cref{fig:blumensath-appendix} for examples. In particular, the alternation rank of the formulas $\phi(x,\bar y)$ with a single free variable $x$ over every $\Omega$-indiscernible sequence in a monadically NIP structure is at most $2$.

\begin{figure}[H]
  \begin{center}
  \includegraphics[scale=1]{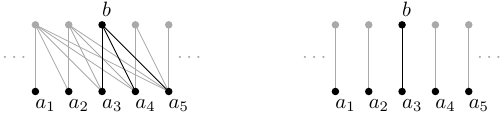}
  \end{center}
  \caption{Two monadically NIP structures. On the left: the infinite half-graph, where \cref{thm:blumensath} applies for the edge relation with $t_<(b) = 0$ and $t_>(b) = 1$.
  On the right: the infinite matching, where
  we have $t_<(b)  = t_>(b) = 0$ but a differing truth value at index $\mathrm{ex}(b) = 3$.
  }
  \label{fig:blumensath-appendix}
\end{figure}

\cref{thm:blumensath} will be a crucial tool for proving flip-flatness. 
However, as we strive for algorithmic applications, we have to develop an effective, computable version that is suitable for handling classes of finite structures.
Instead of requiring $\Omega$-indiscernibility, we will specifically consider $\Delta$-indiscernible sequences with respect to special sets $\Delta=\Delta^\Phi_k$ that we define soon.
These sets $\Delta^\Phi_k$ will strike the right balance of being
on the one hand sufficiently rich to allow us to derive structure from them, but on the other hand
sufficiently simple such that we can efficiently evaluate formulas from $\Delta^\Phi_k$, making our flip-flatness result algorithmic.

Fix a finite set $\Phi(x,\bar{y})$ of formulas of the form $\phi(x,\bar y)$
where $x$ is a single variable.
A \emph{$\Phi$-pattern} is a finite sequence $(\phi_i)_{1\le i\le k'}$,
where each formula $\phi_i(x,\bar y)$
is a boolean combination of formulas $\phi(x,\bar y)\in \Phi(x,\bar y)$.
 Given a $\Phi$-pattern $(\phi_i)_{i\in [k']}$,
the following formula expresses that, 
for a given sequence 
of $k'$ tuples, each of length $|\bar y|$, there is some element that realizes that pattern (see \cref{fig:pattern-appendix} for an example):
\[\gamma_{(\phi_1,\ldots,\phi_{k'})}(\bar y_1,\ldots,\bar y_{k'})
=\exists x.\bigwedge_{i\in [k']}\phi_i(x,\bar y_i).\]

For a finite set of formulas $\Phi(x,\bar{y})$ and an integer $k$
we define $\Delta_k^{\Phi}$ to be the set of 
all formulas $\gamma_{(\phi_1,\ldots,\phi_{k'})}$,
where $(\phi_1,\ldots,\phi_{k'})$ is a $\Phi$-pattern of length $k' \leq k$. Note that the set~$\Delta_k^{\Phi}$ is finite, as (up to equivalence) there are only finitely many boolean combinations of formulas in $\Phi$.
We write $\Delta_k^{\phi}$ for $\Delta_k^{\{\phi\}}$.

\begin{figure}[H]
  \begin{center}
  \includegraphics[scale=1]{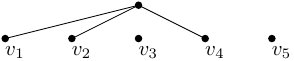}
  \end{center}
  \caption{Example of an $\{E\}$-pattern: a graph satisfying $G\models \gamma_{(E,E,\neg E,E,\neg E)}(v_1,v_2,v_3,v_4,v_5)$.}
  \label{fig:pattern-appendix}
\end{figure}

We can now state a finitary version of \cref{thm:blumensath} that we prove in this section.
For the convenient use later~on, we state the result not for single formulas \(\phi\) but for \(\Phi\)-types.

\begin{restatable}{theorem}{iseqNMonNipMultipleFormulas}\label{thm:iseq_n_mon_nip_multiple_formulas}
  For every monadically NIP class $\CC$ of structures and finite set of formulas $\Phi(x,\bar y)$ 
  there exist integers $n_0$ and $k$, such that for every 
  $\struc A \in \mathscr{C}$ the following holds. 
  If $I = (\bar a_i)_{1\leq i\leq n}$ is a $\Delta^\Phi_k$-indiscernible sequence of length $n\geq n_0$ in $\struc A$,
  and $b \in \struc A$, 
  then there exists an exceptional index $\mathrm{ex}(b) \in [n]$ and 
  $\Phi$-types $\tau_-(x,\bar y)$ and $\tau_+(x,\bar y)$
  such that
  \begin{alignat*}{2}
      \struc A\models \tau_-(b,\bar a_i) \text{ holds for all }& 1\le i < \mathrm{ex(b)}, \text{ and }
      \\
      \struc A\models \tau_+(b,\bar a_i) \text{ holds for all }& \mathrm{ex}(b)<i\le n.
  \end{alignat*}
\end{restatable}

Our proof uses different tools than~\cite{blumensath2011}. It is combinatorial, fully constructive, and gives explicit bounds on $n_0$ and $k$ as well as the required set of formulas $\Delta^\Phi_k$.
One may alternatively finitize \cref{thm:blumensath} via compactness,
but then we do not obtain these crucial properties.

For the more restricted monadically stable classes we can give even stronger guarantees: for every element $b\in\struc A$, the types do not alternate and we have $\tau_- = \tau_+$.

\begin{restatable}{theorem}{iseqNMonStableMultipleFormulas}
  \label{cor:iseq_n_mon_stable_multiple_formulas}
  For every monadically stable class $\CC$ of structures and finite set of formulas~$\Phi(x,\bar y)$
  there exist integers $n_0$ and $k$, such that for every
  $\struc A \in \mathscr{C}$ the following holds.
  If $I = (\bar a_i)_{1\leq i\leq n}$ is a $\Delta^\Phi_k$-indiscernible sequence of length $n\geq n_0$ in $\struc A$,
  and $b \in \struc A$,
  then there exists an exceptional index $\mathrm{ex}(b) \in [n]$ and a $\Phi$-type $\tau$, such that
  \[\struc A\models \tau(b,\bar a_i)
  \qquad\text{for all $i\in [n]$ with
  $i\neq\mathrm{ex}(b)$.}
  \]
\end{restatable}


\subsection{Alternation rank in monadically NIP classes}\label{apx:alternation-rank}

%
%
%

We start by showing that in monadically NIP classes, for every formula $\phi(x,\bar y)$ there exists~$k$ such that the alternation rank of $\phi$ over a sufficiently long $\Delta^\phi_k$-indiscernible sequence is bounded by $2$.

\begin{lemma}\label{lem:alternation_rank}
  For every monadically NIP class $\CC$ of $\Sigma$-structures and formula $\phi(x,\bar y)$
  there exist integers~$n_0$ and $k$, such that for every 
  $\struc A \in \mathscr{C}$ the following holds.
  If $I = (\bar a_i)_{1\leq i\leq n}$ is a $\Delta^\phi_k$-indiscernible sequence of length $n\geq n_0$ in $\struc A$, then $\phi$ has alternation rank at most $2$ over~$I$.
\end{lemma}

\begin{proof}
As $\Cc$ is monadically NIP, by \cref{cor:pairing-index} there exists a minimal $k'$ such that 
$\phi_\oplus(x,\bar y_1 \bar y_2) := \phi(x, \bar y_1) \oplus \phi(x, \bar y_2)$
has pairing index less than $k'$ in $\CC$, where $\oplus$ denotes the XOR operation. 
Let $k \coloneqq 2k'$ and $n_0\coloneqq 20k'$. 

  Assume towards a contradiction that $\phi$ has alternation rank at least $3$.
  Then without loss of generality there exist indices $1 \leq i_1 < i_2 < i_3 < i_4 \leq n$
  and an element $a\in \struc A$ that is $\phi$-connected to $\bar a_{i_1}, \bar a_{i_3}$
  but not $\bar a_{i_2}, \bar a_{i_4}$.
  The formula
  \[
      \gamma_{(\phi,\neg \phi,\phi,\neg \phi)}(\bar x_1,\bar x_2,\bar x_3,\bar x_4) 
    =
    \exists z.\big(
    \phi(z,\bar x_1) 
    \wedge 
    (\neg \phi)(z,\bar x_2) 
    \wedge 
    \phi(z,\bar x_3) 
    \wedge 
(\neg\phi)(z,\bar x_4)\big)
  \]
  is contained in $\Delta^\phi_k$.
  Since $\gamma_{(\phi,\neg \phi,\phi,\neg \phi)}(\bar a_1,\bar a_2,\bar a_3,\bar a_4)$ holds in $\struc A$, as witnessed by $a$,
  by $\Delta^\phi_k$-indis\-cer\-nibility of $I$ and since $n \ge 20k'$, there must also exist an element $b \in \struc A$
  witnessing the truth~of
  \[\struc A \models 
  \gamma_{(\phi,\neg \phi,\phi,\neg \phi)}(
    \bar a_{j_1 := 4k'},
    \bar a_{j_2 := 8k'},
    \bar a_{j_3 := 12k'},
    \bar a_{j_4 := 16k'}
  ).
  \]
  
  By a simple majority argument there exists a (possibly non-consecutive) subsequence~$M$ of~$I$ of length exactly $2k'$ that is located between $\bar a_{j_2}$ and $\bar a_{j_3}$, and to which $b$ is homogeneously $\phi$-connected (that is, $b$ is $\phi$-connected to either all or no elements of $M$).

  Assume $b$ is $\phi$-connected to no element in $M$. We set $\ell \coloneqq j_1$ and $r \coloneqq j_3$.
  Again by applying majority arguments we find a subsequence $L$ of $I$ of length exactly~$2k'$ before~$\bar a_\ell$ to which~$b$ is homogeneously $\phi$-connected and a subsequence $R$ of $I$ of length exactly~$2k'$ after~$\bar a_r$ to which $b$ is homogeneously $\phi$-connected.
  The possible cases of how $b$ is $\phi$-connected towards the subsequence $(L,\bar a_\ell, M,\bar a_r, R) \subseteq I$
  are depicted in \cref{fig:neighbourhood_nip}.
  
  \begin{figure}[h!]
    \begin{center}
    \includegraphics[scale=0.75]{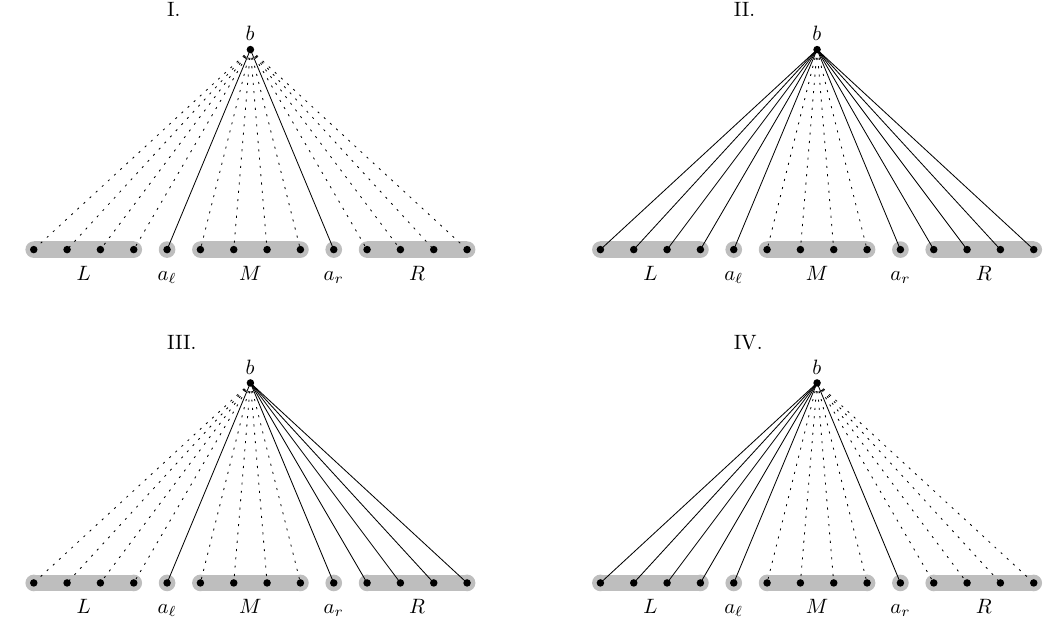}
    \end{center}
    \caption{A case distinction on the connections between $b$ and $L$ and $R$.
    Solid lines denote $\phi$-connections while
    dashed lines denote $\neg\phi$-connections. The cases when $b$ is $\phi$-connected to $M$ are symmetric. }
    \label{fig:neighbourhood_nip}
\end{figure}
  
We regroup the elements in $(L,\bar a_\ell, M,\bar a_r, R)$ 
by merging each tuple with an odd index with respect to $(L,\bar a_\ell, M,\bar a_r, R)$ with its successor,
yielding a new sequence $J = (L',\bar a'_\ell, M', \bar a'_r, R')$
where $L'$ and $R'$ are each of length exactly $k'$, 
$\bar a'_\ell := \bar a_\ell \bar a_{\ell+1}$ and 
$\bar a'_r := \bar a_{r-1} \bar a_{r}$ and $M'$ is of length exactly $k'-1$.
As depicted in \cref{fig:neighbourhood_nip_xor}, $b$ is $\phi_\oplus$-connected only to $\bar a'_\ell$ and $\bar a'_r$ among the tuples in $J$.

\begin{figure}[h!]
  \begin{center}
  \includegraphics[scale=0.75]{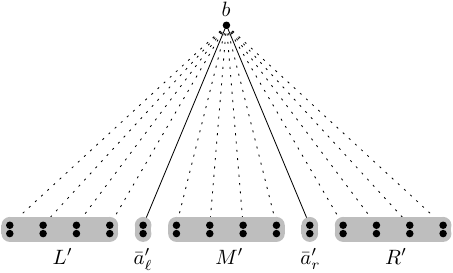}
  \end{center}
  \caption{A visualization of the regrouped sequence $J$.
  Solid lines denote $\phi_\oplus$-connections while
  dashed lines denote $\neg\phi_\oplus$-connections.}
  \label{fig:neighbourhood_nip_xor}
\end{figure}

Therefore, by omitting elements from $L',M',R'$ we find for every pair of distinct indices $i, j \in [k']$ a subsequence $J_{ij} \subseteq J$ such that $b$ is $\phi_\oplus$-connected exactly to the $i$th and $j$th tuple in $J_{ij}$, which is expressed by the formula
\[
  \gamma_{ij}(\bar x_{1}, \bar y_{1} ,\ldots,\bar x_{k'}, \bar y_{k'}) = \exists z.
  \big(\bigwedge_{\ell \in \{i,j\}} \phi_\oplus(z,\bar x_{\ell} \bar y_{\ell}) \hspace{3mm}\wedge \bigwedge_{\ell \in [k']\setminus\{i,j\}} \hspace{-3mm}\neg\phi_\oplus(z,\bar x_{\ell} \bar y_{\ell})\big).
\]

If we are able to show that $\gamma_{ij}$ is contained in $\mathrm{EM}(I)$
for all $i, j \in [k']$,
then any subsequence of length~$2k'$ of $I$ witnesses that $\phi_\oplus$ has pairing index at least $k'$ and we arrive at a contradiction as desired.
Fix $i, j \in [k']$. 
Unpacking the tuples of $J_{ij}$ we obtain a subsequence $I_{ij}$ of $I$ on which $\gamma_{ij}$ is true.
However, since $\gamma_{ij}$ is not contained in $\Delta^\phi_k$, we cannot immediately conclude that $\gamma_{ij} \in \mathrm{EM}(I)$.

Instead, let $S \in \{\phi, \neg \phi\}^{2k'}$ be the $\phi$-connection pattern of $b$ towards $I_{ij}$.
This pattern depends on the way $b$ is connected towards $L,M,R$ as seen in \cref{fig:neighbourhood_nip}.
For example if $b$ is connected to all elements of $L$, no element of $M$ and all elements of $R$, we have that
\[
S=(\underbrace{\underbrace{\underbrace{\textcolor{blue}{\phi,\ldots,\phi}}_{2i-1},\textcolor{red}{\neg\phi,\ldots,\neg\phi}}_{2j-1},\textcolor{blue}{\phi,\ldots,\phi}}_{2k'}).
\]
The formula $\gamma_S(\bar x_1,\ldots, \bar x_{2k'})$ (defined via $S$ at the beginning of the section) expresses that there exists an element whose $\phi$-connection pattern towards the sequence $(\bar x_1,\ldots, \bar x_{2k'})$ is described by $S$.
The formula holds true on the sequence $I_{ij}$, as witnessed by $b$.
By definition $\gamma_S$ is contained in $\Delta_k^\phi$, and it follows by $\Delta_k^\phi$-indiscernibility of $I$ that $\gamma_S$ is contained in $\mathrm{EM}(I)$.
Since $\gamma_S$ implies $\gamma_{ij}$, we have that also $\gamma_{ij}$ is contained in $\mathrm{EM}(I)$ and we arrive at the desired contradiction.

Previously, we assumed $b$ is not $\phi$-connected to $M$.
If $b$ is $\phi$-connected to $M$, the same proof works using $\ell := j_2$ and $r := j_4$.
\end{proof}

The bound of $2$ is tight, since already in the class of matchings, we find arbitrarily long indiscernible sequences over which the edge relation has alternation rank $2$.
To achieve this bound, the restriction to formulas whose first free variable is a singleton is necessary.
Formulas with larger $\bar x$-tuples as free variables can have higher 
alternation rank even in monadically stable classes.
Consider for example the formula $\phi(x_1x_2,y) = (x_1 = y) \vee (x_2 = y)$, which has alternation rank $4$ on the class of edgeless graphs. 

Our following central observation shows that not only the alternation rank is bounded by~$2$ in monadically NIP classes, but in fact the alternations are centered around a single exceptional index.

\begin{restatable}{lemma}{iseqmonnip}
  \label{thm:iseq_n_mon_nip}
  For every monadically NIP class $\CC$ of structures  and formula $\phi(x,\bar y)$ 
  there exist integers~$n_0$ and $k$, such that for every 
  $\struc A \in \mathscr{C}$ the following holds. 
  If $I = (\bar a_i)_{1\leq i\leq n}$ is a $\Delta_k^\phi$-indiscernible sequence of length $n\geq n_0$ in $\struc A$,
  and $a \in \struc A$, then there exists an exceptional index $\mathrm{ex}(a) \in [n]$ and 
  two truth values $t_<(a),t_>(a) \in \{0,1\}$ such that
    \begin{align*}
      &\text{for all } 1 \leq i < \mathrm{ex(a)}: 
      \struc A \models \phi(a,\bar a_i) \hspace{2mm} \iff \hspace{2mm} t_<(a) = 1, \text{ and } \\
      &\text{for all } \mathrm{ex(a)} < i \leq n: 
      \struc A \models \phi(a,\bar a_i)\hspace{2mm} \iff \hspace{2mm} t_>(a)=1.
    \end{align*}
\end{restatable}

\begin{proof}
Choose $n_0,k$ and $k'$ as in the proof of \cref{lem:alternation_rank}.
If $\phi$ has alternation rank less than~$2$ in $I$, then the statement
obviously holds. 
Therefore, by \cref{lem:alternation_rank} we may assume that $\phi$ has alternation rank exactly~$2$.

Assume towards a contradiction that there exists an element $a \in \struc A$ whose
$\phi$-connection to every element of~$I$ cannot be described by choosing an exceptional index and two truth values. 
This is witnessed by indices $1 \leq i_1<i_2<i_3<i_4 \leq n$ such that $a$ is either $\phi$-connected to~$\bar a_{i_1}$ and $\bar a_{i_4}$ and not $\phi$-connected to $\bar a_{i_2}$ and $\bar a_{i_3}$ or vice versa.
We only consider the first case as the second follows by symmetry.
Note that the formula
\[
\gamma_{(\phi,\neg\phi,\neg\phi,\phi)}(\bar x_1,\bar x_2,\bar x_3,\bar x_4) 
:=
\exists z.\big(
\phi(z,\bar x_1) 
\wedge 
(\neg \phi)(z,\bar x_2) 
\wedge 
(\neg \phi)(z,\bar x_3) 
\wedge 
\phi(z,\bar x_4)\big)
\]
is contained in $\Delta^\phi_k$ and thus, as witnessed by $a$, is contained in $\mathrm{EM}(I)$.
Since additionally the alternation rank of $\phi$ is exactly $2$, there exists an element $b \in \struc A$
that is $\phi$-connected to exactly the first $2k'$ and the last~$2k'$ elements of $I$ (and to none of the elements in between).
As in the proof of \cref{lem:alternation_rank},
we can now regroup the elements in $I$ into a new sequence $J = (L',\bar a'_\ell, M', \bar a'_r, R')$ and
argue towards a contradiction in exactly the same way.
\end{proof}

In order to prove \cref{thm:iseq_n_mon_nip_multiple_formulas}, which is restated below, we generalize \cref{thm:iseq_n_mon_nip} to sets of formulas $\Phi$ by showing that there exists an integer~$k$ such that over $\Delta_k^\Phi$-indiscernible sequences for all $a\in \struc A$ all formulas of $\Phi$ have the same exceptional index. 

Recall that for a finite set $\Phi(x,\bar{y})$ of formulas of the form $\phi(x,\bar y)$,
where $x$ is a single variable,
a \emph{$\Phi$-type} is 
a conjunction $\tau(x,\bar y)$ of formulas in $\Phi$
or their negations, such that every formula in $\Phi$
occurs in $\tau(x,\bar y)$ either positively or negatively.

\iseqNMonNipMultipleFormulas*
\begin{proof}
  Choose $n_0$ and $k$ such that for every $\Phi$-type $\tau(x,\bar y)$,
  \cref{thm:iseq_n_mon_nip} holds on $\CC$ with these values of~$n_0$ and $k$.
  Let $I = (\bar a_i)_{1\leq i\leq n}$ be a $\Delta^\Phi_k$-indiscernible sequence of length $n\geq n_0$ and 
   let $a \in \struc A$.
   Note that $I$ is also $\Delta^\tau_k$-indiscernible  for every $\Phi$-type $\tau$.

  By increasing $n_0$ if necessary, we may suppose that $n_0$ is greater than $2^{|\Phi|}$,
  in particular, it is greater than the number of $\Phi$-types. 
  Thus, 
  there exists a $\Phi$-type $\tau$ such that $a$ is $\tau$-connected to
  at least two elements $\bar a_{\ell_1}$ and $\bar a_{\ell_2}$ from $I$.
  If $a$ is $\tau$-connected to all but at most one element from $I$, then we are done. 
  Hence, assume that $a$ is not $\tau$-connected to at least two elements from $I$.
  Since~$I$ is $\Delta_k^{\tau}$-indiscernible,
  by \cref{thm:iseq_n_mon_nip} there exists an index $i\in[n]$ such that 
  either $a$ is $\tau$-connected to the elements $(\bar a_j)_{j< i}$ 
  and not $\tau$-connected to the elements $(\bar a_j)_{j\geq i}$, or vice versa.
  By symmetry, we can assume the first case.
  Therefore, $\tau$ will play the role of $\tau_-$ in the statement. We now prove that there is a $\Phi$-type $\tau_+$ such that $a$ is $\tau_+$-connected to all tuples after $\bar a_i$ in $I$.
  
  Assume towards a contradiction that there are two distinct $\Phi$-types $\sigma_1,\sigma_2$
  and indices $r_1,r_2\in[n]$ greater than $i$, such that $a$ is $\sigma_1$-connected to $\bar a_{r_1}$ 
  and $\sigma_2$-connected to $\bar a_{r_2}$.
  By symmetry, we can assume $1 \leq \ell_1 < \ell_2 < i < r_1 < r_2 \leq n$.
  Now, if $a$ is $\sigma_1$-connected to $\bar a_i$, then we have
  \[
    \struc A \models 
    (\neg \sigma_1)(a, \bar a_{\ell_2}) 
    \wedge 
    \sigma_1(a, \bar a_{i})
    \wedge
    \sigma_1(a, \bar a_{r_1})
    \wedge
    (\neg \sigma_1)(a, \bar a_{r_2}).
  \]
  This however contradicts \cref{thm:iseq_n_mon_nip}, as $I$ is $\Delta_k^{\sigma_1}$-indiscernible and there can be at most one exceptional index for $a$ and $\sigma_1$. 
  
  Otherwise, $a$ is $\sigma_3$-connected towards $\bar a_i$ for some $\Phi$-type  $\sigma_3$ which is neither equal to~$\tau$ nor $\sigma_1$.
  Denote $\pi:=\sigma_2\lor\sigma_3$.
  Then we have
  \[
    \struc A \models 
    (\neg\pi)(a, \bar a_{\ell_2}) 
    \wedge 
    \pi (a, \bar a_{i})
    \wedge
    (\neg\pi)(a, \bar a_{r_1})
    \wedge
    \pi(a, \bar a_{r_2}).
  \]
Again this contradicts \cref{thm:iseq_n_mon_nip}, as $I$ is $\Delta_k^{\pi}$-indiscernible and there can be at most one exceptional index for $a$ and $\pi$. 

  It follows that there exists a $\Phi$-type $\tau_+$ such that $a$ is $\tau_+$-connected to all elements after $\bar a_i$ in~$I$.
  Moreover, $a$ is $\tau$-connected to all elements before $\bar a_i$ in $I$.
\end{proof}

Restricting to the monadically stable case, we can finally prove \cref{cor:iseq_n_mon_stable_multiple_formulas}, which is restated below for convenience.

\iseqNMonStableMultipleFormulas*
\begin{proof}
  Let $\tau_-$ and $\tau_+$ be the $\Phi$-types obtained from 
    \Cref{thm:iseq_n_mon_nip_multiple_formulas}, so that $a$ is $\tau_-$-connected to all elements of $I$ before the exceptional one, and is $\tau_+$-connected to all elements of $I$ after the exceptional one.
    Suppose towards contradiction that it is not the case that all but one element of $I$ 
    is $\tau_-$ connected to $a$,
    and it is not the case that all but one element of $I$ is $\tau_+$ connected to $a$.
    In particular, $\tau_-$ and $\tau_+$ are distinct,
    and there are indices $i<j$ in $[n]$ with 
  \[
    \struc A \models
    \tau_-(a,\bar a_1) \wedge
    \tau_-(a,\bar a_i) \wedge
    (\neg\tau_-)(a,\bar a_j)\wedge
    (\neg\tau_-)(a,\bar a_n),
  \]
  or 
  \[
    \struc A \models
    (\neg\tau_+)(a,\bar a_1) \wedge
    (\neg\tau_+)(a,\bar a_i) \wedge
    \tau_+(a,\bar a_j)\wedge
    \tau_+(a,\bar a_n).
  \]
  By symmetry, assume the first case holds.
By indiscernibility of $I$, 
for every pair $i<j$ of indices in $[n]$,
we have 
\[\struc A \models\exists z.
\tau_-(z,\bar a_1) \wedge
\tau_-(z,\bar a_i) \wedge
(\neg\tau_-)(z,\bar a_j)\wedge
(\neg\tau_-)(z,\bar a_n).\]
Moreover, for any pair $1<j\le i<n$ of indices in $[n]$
the above formula does \emph{not} hold in $\struc A$:
for $j=i$ it is clearly contradictory, and for $j<i$
this would yield alternation rank $3$,
contrary to \cref{lem:alternation_rank}.
Hence, the formula 
  \[\gamma(\bar x_1,\bar x_2)=
    \exists z.
    \tau_-(z,\bar a_1) \wedge
    \tau_-(z,\bar x_1) \wedge
    (\neg\tau_-)(z,\bar x_2)\wedge
    (\neg\tau_-)(z,\bar a_n).
  \]
  (with two tuples of variables as free variables and~$\bar a_1$ and~$\bar a_n$ fixed as parameters)
  defines an order on $\bar a_{2}, \dots \bar a_{n-1}$. 
  Since $\CC$ is stable, there is an integer $\ell$ such that $\psi$ cannot define orders longer than $\ell$. By choosing $n_0$ greater than $\ell+2$ we obtain a contradiction. 
\end{proof}

The following statement
 follows immediately from \Cref{cor:iseq_n_mon_stable_multiple_formulas}.

\begin{corollary}
  \label{thm:iseq_n_mon_stable}
  For every monadically stable class $\CC$ of structures  and formula $\phi(x,\bar y)$ 
  there exist integers $n_0$ and $k$, such that for every 
  $\struc A \in \mathscr{C}$ the following holds. 
  If $I = (\bar a_i)_{1\leq i\leq n}$ is a $\Delta_k^\phi$-indiscernible sequence of length $n\geq n_0$ in $\struc A$,
  and $a \in \struc A$, then either
  \[|N_\phi(a) \cap I| \leq 1 \text{ or } |N_\phi(a) \cap I| \geq |I|-1.\]
\end{corollary}

\section{Disjoint definable neighborhoods in monadically NIP classes}
\label{sec:families}

In the previous section we have seen that in monadically stable and monadically NIP classes, every element is very homogeneously connected to all but at most one element of an indiscernible sequence.
At a first glance, this exceptional behavior towards one element seems to be erratic and standing in the way of combinatorial and algorithmic applications.
However, it turns out that it can be exploited to obtain additional structural properties.

The key observation is that elements that are ``exceptionally connected'' towards a single element of an indiscernible sequence, inherit some of the good properties of that sequence.
We will give a simple example to demonstrate this idea.
Let $G$ be a graph containing certain red vertices $R$ and blue vertices $B$, such that the edges between $R$ and $B$ describe a matching (see \cref{fig:matching_stars-appendix}, left).

\begin{figure}[h]
  \begin{center}
  \includegraphics[scale=1]{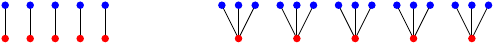}
  \end{center}
  \caption{Examples of one-to-one and many-to-one connections in a graph. On the left: a matching. On the right: a star forest.}
  \label{fig:matching_stars-appendix}
\end{figure}

Given a formula $\phi(x_1,\ldots, x_q)$, define 
\[
\hat \phi(x_1,\ldots,x_q)
:= 
\exists b_1,\ldots,b_q \in B. 
\phi(b_1,\ldots,b_q)
\wedge
\bigwedge_{i\in [q]} E(x_i, b_i).
\]

Take a $\hat \phi$-indiscernible sequence $R'$ among $R$.
Obviously, every vertex in the blue neighborhood $B'$ of $R'$ has one exceptional connection towards $R'$, that is, towards its unique matching neighbor in~$R'$.
It is now easy to see that $B'$ is a $\phi$-indiscernible sequence in $G$: for every sequence of red vertices $a_1,\ldots,a_q \in R'$
and their unique blue matching neighbors $b_1,\ldots,b_q \in B'$ we have
$
G\models \hat \phi (a_1,\ldots,a_q )$ if and only if 
$G\models \phi (b_1,\ldots,b_q)$.

This example sketches how first-order definable one-to-one connections
towards elements of an indiscernible sequence preserve indiscernibility.
The more general case however is the many-to-one case,
where we have a set of elements, each of which is exceptionally connected to a single element of an indiscernible sequence $I$, while allowing multiple elements to be exceptionally connected to the same element of $I$
(think, for example, of $R$ and $B$ being the centers and leaves of a star-forest as depicted in \cref{fig:matching_stars-appendix}, right).
Our notion of such many-to-one connections will be that of \emph{disjoint $\alpha$-neighborhoods}.
Recall, that for a formula~$\alpha(\bar x,y)$, the \(\alpha\)-neighborhood $N_\alpha(\bar a)$ of a tuple $\bar a$ is defined as the set of all \(b\) satisfying $\alpha(\bar a,b)$.
We say a sequence $J$ of  $|\bar x|$-tuples has \emph{disjoint $\alpha$-neighborhoods} if 
$N_\alpha(\bar a_1)\cap N_\alpha(\bar a_2)=\varnothing$ for all distinct $\bar a_1,\bar a_2\in J$.

As the main technical tool of this paper, we prove for
monadically NIP classes
that every sequence of disjoint $\alpha$-neighborhoods contains a large subsequence that exerts strong control over its neighborhood.
This lifts the strongly regular behaviour of idiscernible sequences to sequences of disjoint $\alpha$-neighborhoods.
The main result of this section, \cref{thm:disjoint_families_nip}, is the backbone of our flip-flatness proof and states the following.
In every large sequence of disjoint $\alpha$-neighborhoods, we will find 
a still-large subsequence such that the $\phi$-connections of every element $a$ towards all $\alpha$-neighborhoods in the subsequence can be described by a bounded set of sample elements in the following sense.
After possibly omitting one exceptional neighborhood at index $\mathrm{ex(a)}$ depending on~$a$,
the \(\phi\)-connections are completely homogeneous (in the stable case) or alternate at most once (in the NIP case).
A related (but orthogonal) result was also proved in \cite[Lemma 64]{twwIVarxiv} using tools from model theory.

\begin{figure}[h]
  \begin{center}
  \includegraphics[scale=1.2]{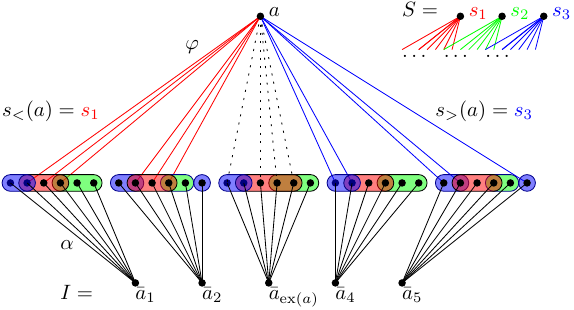}
  \end{center}
  \caption{
    A visualization of \cref{thm:disjoint_families_nip}. On the bottom: a sequence $I$ with disjoint $\alpha$-neighborhoods. On the top right: a small set of sample elements $S$.
    The $\phi$-neighborhoods of the elements in $S$ in the $\alpha$-neighborhood of $I$ are colored accordingly. Note that their $\phi$-neighborhoods can overlap.
The $\phi$-neighborhood of $a$ in the $\alpha$-neighborhood of $I$ is equal to the $\phi$-neighborhood of $s_<(a) = s_1$ before the exceptional index $\mathrm{ex}(a)$. After it, it is equal to the $\phi$-neighborhood of $s_>(a) = s_3$. For the $\alpha$-neighborhood of $\bar a_{\mathrm{ex}(a)}$ we make no claims.}
  \label{fig:families-appendix}
\end{figure}

\begin{restatable}{theorem}{thmDisjointFamiliesNip}\label{thm:disjoint_families_nip}
  For every monadically NIP class of structures $\Cc$, and formulas $\phi(x,y)$ and $\alpha(\bar x,y)$, there exists a function $N:\NN \rightarrow \NN$ and an integer $k$ such that 
  for every $m\in \NN$
  every structure $\struc A \in \CC$ and every finite sequence of tuples $J\subseteq \struc A^{|\bar x|}$ of length $N(m)$ whose $\alpha$-neighborhoods are disjoint the following holds. 
  There exists a subsequence $I =(\bar a_i)_{1\leq i\leq m}\subseteq J$ of length $m$ and a 
  set $S \subseteq \struc A$ of at most $k$ sample elements such that for every element $a \in \struc A$ there exists an exceptional index $\mathrm{ex}(a) \in [m]$ and a pair $s_\sless(a),s_\sgreater(a)\in S$ such that 
  \begin{alignat*}{2}
      \text{for all }& ~1 \leq i < \mathrm{ex(a)}&&: N_\alpha(\bar a_i) \cap N_\phi(a) = N_\alpha(\bar a_i) \cap N_\phi(s_\sless(a)), \text{ and} \\
      \text{for all }& ~\mathrm{ex(a)} < i \leq m&&: N_\alpha(\bar a_i) \cap N_\phi(a) = N_\alpha(\bar a_i) \cap N_\phi(s_\sgreater(a)).
  \end{alignat*}
  If $a \in N_\alpha(\bar a_i)$ for some $\bar a_i\in I$, then $i = \mathrm{ex}(a)$.
  
  \medskip\noindent
  If $\CC$ is monadically stable, then $s_\sless(a) = s_\sgreater(a)$ for every $a\in \struc A$. 
\end{restatable}

A visualization of the NIP case of \cref{thm:disjoint_families_nip} is provided in \cref{fig:families-appendix}.
After proving \Cref{thm:disjoint_families_nip}, we will give an algorithmic version of the result in \cref{sec:algorithmic-aspects}. 
We will need the following definition for both.


\begin{definition}\label{def:phi_equality}
For a set $B$ we call two elements \emph{$\phi$-equivalent over $B$} if they 
have the same $\phi$-neighborhood in $B$ and either both or none of them are contained in $B$. 
Respectively, we call them \emph{$\phi$-inequivalent over $B$} if they are not $\phi$-equivalent over $B$.
Observe that there exists a formula
\[
  \mathrm{eq}^\phi_\alpha(x,y,\bar z)
  := 
  \big (\alpha(\bar z, x) \leftrightarrow \alpha(\bar z, y) \big)
  \wedge
  \forall w. \Big(\alpha(\bar z, w) 
  \rightarrow 
  \big (\phi(x,w) \leftrightarrow \phi(y,w) \big )\Big)
\]
that is true if and only if $x$ and $y$ are $\phi$-equivalent over the $\alpha$-neighborhood of $\bar z$. 
\end{definition}
%
%
%

\begin{proof}[Proof of \cref{thm:disjoint_families_nip}]
    We split the proof into multiple paragraphs.
        
    \medskip\noindent\textbf{Defining constants $\boldsymbol \ell$, $\boldsymbol s$ and $\boldsymbol k$.}
    We are going to use a Ramsey-type result for set systems due to Ding et al.~\cite{DBLP:journals/jct/DingOOV96},
    which we are going to reformulate using notation from graph theory for convenience.
    We say a bipartite graph with sides $a_1,\ldots, a_\ell$ and $b_1,\ldots, b_\ell$ forms
    \begin{itemize}
    \item a \emph{matching} of order $\ell$ if $a_i$ and $b_j$ are adjacent if and only if $i=j$ for all $i,j\in [\ell]$, 
    \item a \emph{co-matching} of order $\ell$ if $a_i$ and $b_j$ are adjacent if and only if $i\neq j$ for all $i,j\in [\ell]$, 
    \item a \emph{ladder} of order $\ell$ if $a_i$ and $b_j$ are adjacent if and only if $i\leq j$ for all $i,j\in [\ell]$. 
    \end{itemize}

    \begin{restatable}[{\cite[Corollary 2.4]{DBLP:journals/jct/DingOOV96}}]{theorem}{thmmatching}\label{thm:matching}
      There exists a function $Q:\NN\rightarrow\NN$ such that for every $\ell\in\NN$ and for every bipartite graph $G = (L,R,E)$ without twins, where $L$ has size at least~$Q(\ell)$, contains a matching, co-matching, or ladder of order $\ell$ as an induced subgraph.
  \end{restatable}

    Let $q$ be the maximum of the quantifier rank of $\alpha$ and $\phi$.
    Let $\Sigma$ be the signature of $\CC$ and for $i\in \NN$ let~$\Sigma^+_i$ be a monadic extension of $\Sigma$ with $i$ new colors. 
    Let $\ell$ be the smallest number such that every formula over~$\Sigma^+_3$ with the same free variables as $\alpha$ and quantifier rank at most $q+3$ has pairing index less than $\ell$
    in $\Cc[\Sigma^+_3]$. Since up to equivalence there are only finitely many such formulas and $\CC$ is monadically NIP, by \cref{cor:pairing-index} such $\ell$ exists. 
    Let $s := Q^{3\ell}(\ell^2)$, where $Q$ is the function given by \cref{thm:matching}.
    Let $\{c_1,\ldots,c_s\}$ be a set of $s$ fresh constant symbols.
    For every $i\in[s]$  define 
    \[
    \Phi_i(x,\bar y) 
    := 
    \{\mathrm{eq}_\alpha^\phi(c_j,x,\bar y) ~:~ j\in [i]\}
    .\]

    Choose $k$ larger than $s$ such that
    \cref{thm:iseq_n_mon_nip_multiple_formulas} holds for the set of formulas $\Phi_{s}$ with this choice of $k$ 
    in the class of all possible extensions of structures in $\CC$ with constants $\{c_1,\ldots,c_s\}$. 
    
    \newcommand{\aex}{\bar a_{\mathrm{ex}}}
    \newcommand{\Aex}{N_\alpha(\aex)}

    \paragraph*{Construction of $\boldsymbol S$ and $\boldsymbol I$.}
    For a sequence $I$ denote by $\mathcal N_\alpha(I)$ the sequence of its $\alpha$-neighborhoods $(N_\alpha(\bar a))_{\bar a \in I}$.
    For every $i\in[s]$ we are going to inductively construct 
    a set of sample elements $S_i = \{s_1,\ldots,s_i\} \subseteq \struc A$ together with a $\Delta^{\Phi_i}_{k}$-indiscernible subsequence $I_i$ of $J$ in the structure~$\struc A^+_i$, which is the expansion of $\struc A$ with the constant symbols $\{c_1,\ldots,c_i\}$ interpreted by the elements of~$S_i$.
    During the construction we are going to ensure that
    \begin{itemize}
        \item no element of $S_i$ is contained in a set from $\mathcal N_\alpha(I_i)$, and 
        \item all elements from $S_i$ are pairwise $\phi$-inequivalent over every set of $\mathcal N_\alpha(I_i)$.
    \end{itemize}
    
    Start with $S_0 := \varnothing$ and let $I_0$ be a $\Delta^{\Phi_0}_{k}$-indiscernible subsequence of $J$ in $\struc A^+_0 = \struc A$.
    In order to construct~$S_{i+1}$ from $S_{i}$ we proceed as follows.
    If for every $a \in \struc A$ there exist $s_<,s_> \in S_i$ and an element $\aex \in I_i$ and $a$ is $\phi$-equivalent to $s_<$ over all sets in $\mathcal N_\alpha(I_i)$ before $N_\alpha(\aex)$ and $\phi$-equivalent to $s_>$ over all sets in~$\mathcal N_\alpha(I_i)$ after $N_\alpha(\aex)$, 
    then we terminate the induction and set $S := S_{i}$ and $I := I_{i}$.
    Otherwise, we claim the following:

    \begin{claim}
    If the induction does not terminate after step $i$, then there exists an element $s_{i+1} \in \struc A$
    such that over all but at most two sets from $\mathcal{N}_\alpha (I_i)$, $s_{i+1}$ is $\phi$-inequivalent to every element in $S_i$.  
    \end{claim}

    \begin{claimproof}
    Since the induction did not terminate yet, there exists an $a\in \struc A$,
    such that for every element $\aex \in I_i$, either for all $s_< \in S_i$, $a$ is $\phi$-inequivalent to $s_<$ over some set before $\Aex$ in $\mathcal N_\alpha(I_i)$ 
    or for all $s_> \in S_i$, $a$~is $\phi$-inequivalent to $s_>$ over some set after $\Aex$ in $\mathcal N_\alpha(I_i)$.
    By $\Delta_k^{\Phi_i}$-indiscerniblity of $I_i$ we can apply \cref{thm:iseq_n_mon_nip_multiple_formulas}.
    This yields an element $\aex \in I_i$ and $\Phi_i$-types $\tau_-,\tau_0,\tau_+$ 
     such that $a$ is $\tau_-$-connected to all elements before $\aex$,
    $\tau_0$-connected to $\aex$, and 
    $\tau_+$-connected to all elements after $\aex$ in $I_i$.
    This means in particular that on all sets before $\Aex$, $a$ is $\phi$-equivalent to $s_< \in S_i$ if and only if $\mathrm{eq}^\phi_\alpha(s_<,x,\bar y)$ 
    appears as a positive conjunct in $\tau_-$.
    By symmetry, up to reversing $I_i$, we can assume that 
    for all $s_< \in S_i$, $a$ is $\phi$-inequivalent to $s_<$ over some set before $\Aex$.
    By combining these two observations we obtain that
    $\tau_-$ cannot have any positive conjuncts, i.e., on all sets before $\Aex$ in $\mathcal N_\alpha(I_i)$, $a$ is $\phi$-inequivalent to all elements from $S_i$.
    To summarize,
    \begin{itemize}
      \item $\tau_-$ is the conjunction 
      of all conjuncts $\neg\mathrm{eq}^\phi_\alpha(s,x,\bar y)$,
      for all $s\in S_i$,
      \item $a$ is $\tau_-$-connected to all elements before $\aex$ in $I_i$, of which there exists at least one, 
      \item $a$ is $\tau_0$-connected to $\aex$,
      \item $a$ is $\tau_+$-connected to all elements after $\aex$ in $I_i$.
    \end{itemize}
    
    We will now use these observations to find an element $s_{i+1} \in \struc A$ that is $\tau_-$-connected to all but at most two elements from $I_i$, which then proves the claim.
    If $\aex$ is the last element in $I_i$ or if $\tau_+=\tau_-$, 
    then~$a$ is $\tau_-$-connected to all elements of $I_i$ except possibly $\aex$ and we can set $s_{i+1} := a$.
    Assume therefore that there exists an element after $\aex$ in $I_i$ and we have $\tau_-\neq\tau_+$.
    If there exists only a single element before~$\aex$ and $\tau_0 = \tau_+$, 
    then (with $\tau_0=\tau_+ \neq \tau_-$) for some $s_> \in S_i$, we have that $a$ is $\phi$-equivalent to $s_>$ over every but the first set of $\mathcal{N}_\alpha(I_i)$.
    This contradicts our previous choice of $a$.

    Thus, from now on, we know the following two things:
    The first element $\bar b_1$ and the last element $\bar b_3$ of $I_i$ are both not equal to $\aex$.
    Additionally, there either exists a second element $\bar b_2$ between~$\bar b_1$ and $\aex$ in $I_i$, or $\tau_0 \neq \tau_+$.
    The formula 
    \[
    \psi(\bar y_1, \bar y_2, \bar y_3) 
    := 
    \exists z. \big(
    \tau_-(z,y_1) \wedge \neg \tau_+(z,y_2) \wedge  \tau_+(z,y_3)\big)
    \]
    holds on either $(\bar b_1,\bar b_2, \bar b_3)$ in the first case or $(\bar b_1, \aex, \bar b_3)$ in the second case.
    Either way, as $\psi \in \Delta_k^{\Phi_i}$ and by $\Delta_k^{\Phi_i}$-indiscernibility of $I_i$, these witnesses imply that also $\psi \in \mathrm{EM}(I_i)$.
    Therefore, there must also exist an element $s_{i+1}\in \struc A$ that is $\tau_-$-connected to the $(|I_i|-2)$th, not $\tau_+$-connected to the $(|I_i|-1)$th element, and $\tau_+$-connected to the last element of $I_i$.
    By \cref{thm:iseq_n_mon_nip_multiple_formulas}, $s_{i+1}$ is $\tau_-$-connected to all but the last two elements of $I_i$.
  \end{claimproof}
    
  Applying the claim yields an element $s_{i+1}$ that we place into the set $S_{i+1} := S_i \cup \{s_{i+1}\}$.
  Now by removing from $\mathcal{N}_\alpha(I_i)$ the two 
  outlier sets as well as the set containing $s_{i+1}$, we obtain a corresponding sequence $I_i'$ of length at least $|I_i| - 3$ such that the elements of $S_{i+1}$ are pairwise $\phi$-inequivalent over every set in~$\mathcal N_\alpha (I_i')$ and no element from $S_{i+1}$ is contained in any set from $\mathcal N_\alpha (I_i')$.
  We finish the inductive step by setting $I_{i+1}$ to be a $\Delta^{\Phi_{i+1}}_k$-indiscernible subsequence of $I_i'$ in $\struc A^+_{i+1}$.

  If the induction terminates with a set $S$ and a sequence $I=(\bar a_1,\ldots,\bar a_{|I|})$,
  then it is now easy to see that every element $a \in \struc A$ has an exceptional index $1\leq \mathrm{ex}(a) \leq |I|$ and a pair $s_\sless(a),s_\sgreater(a)\in S$ such that 
  \begin{alignat*}{2}
    \text{for all }& ~1 \leq i < \mathrm{ex(a)}&&: N_\alpha(\bar a_i) \cap N_\phi(a) = N_\alpha(\bar a_i) \cap N_\phi(s_\sless(a)), \text{ and} \\
    \text{for all }& ~\mathrm{ex(a)} < i \leq |I|&&: N_\alpha(\bar a_i) \cap N_\phi(a) = N_\alpha(\bar a_i) \cap N_\phi(s_\sgreater(a)).
  \end{alignat*}

    \begin{figure}
      \begin{center}
      \includegraphics[scale=1]{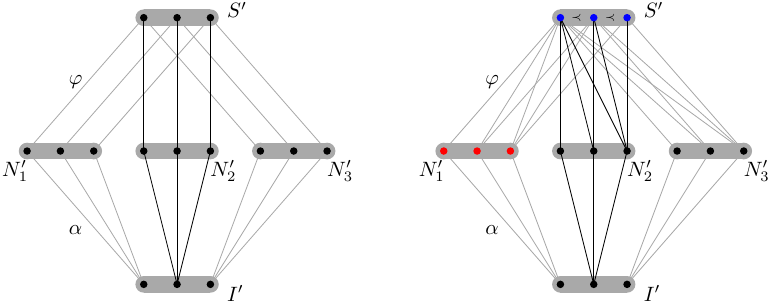}
      \end{center}
      \caption{Visualizations of the matching and the ladder case.}
      \label{fig:families_matching_ladder}
    \end{figure}

    \paragraph*{Bounding the size of $\boldsymbol S$.}
    We claim that the induction terminates after less than $s$ steps, i.e., $S$ has size less than $s$.
    Assume towards a contradiction that $S$ has size at least $s$.
    Take a look at the bipartite graph $\phi$-interpreted between $S$ and any set $F \in \mathcal{N}_\alpha(I)$.
    By construction, all elements from $S$ are pairwise $\phi$-inequivalent over $F$,
    meaning that they either have different $\phi$-neighborhoods in $F$, or their containment in~$F$ differs.
    As no element from $S$ is contained in $F$ we can assume the former
    and thus $S$ contains no twins.
    Since we have that $s = Q^{3\ell}(\ell^2)$ we can iterate \cref{thm:matching} a total of $3\ell$ many times to find a set $S'\subseteq S$ of size $\ell^2$, together with a length $3\ell$ subsequence $J'\subseteq I$, such that for every $\bar a \in J'$ the $\phi$-interpreted bipartite graph between $S'$ and $N_\alpha(\bar a)$ contains an induced matching, co-matching, or ladder of order $\ell^2$.
    By a majority argument, we find a length $\ell$ subsequence $I' \subseteq J'$ such that between $S'$ and every set in $\mathcal{N}_\alpha(I')$, the induced subgraph is of the same type, e.g.\ all matching, all co-matching, or all ladder.
    Let $I' = (\bar a_1,\ldots,\bar a_\ell)$ and for every $i\in[\ell]$ denote by $N'_i$ the restriction of $N_\alpha(\bar a_i)$ to the elements that form via $\phi$ the matching, co-matching or ladder with $S'$.

    Assume we find matchings.
    Let us now show that the formula
    \[
      \psi(x,\bar y) := \exists z \in N_\phi(x) \cap N_\alpha(\bar y).\ R(z)
    \]
    has pairing index at least $\ell$ in the class $\CC$ extended with an additional unary predicate $R$.
    The situation is depicted in the left-hand side of \cref{fig:families_matching_ladder}. 
    The formulas $\phi$ and $\alpha$ interpret a large 1-subdivided biclique in $\struc A$ as follows. 
    The tuples from $I'$ and the elements from $S'$ form the principle vertices.
    The subdivision vertices are formed by $\bigcup_{i\in[\ell]}N_i'$. 
    Due to the assumed matchings, each subdivision element has exactly one incoming $\phi$-connection from $S'$. 
    Due to the tuples in $I'$ having disjoint $\alpha$-neighborhoods, each subdivision element has exactly one incoming $\alpha$-connection from the elements of $I'$. 
    Since $I'$ has size~$\ell$ and $S'$ has size $\ell^2$ we can assign 
    to every pair $\bar a_i,\bar a_j \in I'$ a unique element $s_{ij} \in S'$.
    By marking with the predicate $R$ the two unique subdivision elements from $N_\phi(s_{ij})\cap N_i'$ and $N_\phi(s_{ij})\cap N_j'$, we get that $s_{ij}$ is $\psi$-connected only to $\bar a_i$ and $\bar a_j$ among $I'$.
    As a result, $S'$ and $I'$ witness that $\psi$ has pairing index at least $\ell$.
    
    Let us stress that $\psi$ is works over $\Sigma^+_1$, using only a single new color. The additional constants 
    we marked in $\struc A$ during the construction of $S$ and $I$ were only used to 
    extract $S$ and $I$, but they are 
    not used by $\phi$, $\alpha$ or $\psi$.
    Since $\psi$ has the same free variables as $\alpha$ and quantifier rank at most~$q+1$,
    we obtain a contradiction to our definition of $\ell$.
    
    We have shown that a contradiction arises in case we find $\phi$-matchings between $S'$ and the sets in~$\mathcal N_\alpha(I')$.
    We will now show that we can also build matchings in the cases where we find a co-matching or a ladder.
    Intuitively, matchings, co-matchings and ladders are the graph-equivalent of the comparison relations $=,\neq$, and $\leq$ and our task boils down to expressing $=$ using $\neq$ and $\leq$ respectively.
    In case we find co-matchings, it simply suffices to work with $\neg \phi$ instead of $\phi$.
    In case we find ladders we can recover matchings as follows.
    We introduce two additional color predicates marking the elements in of $S'$ blue and the elements in  $N'_1$ red.
    We can now define a total order on $S'$ using the formula $x_1 \prec x_2 := \exists y. \big( \mathrm{Red}(y) \wedge \phi(x_1,y) \wedge \neg \phi(x_2,y)\big)$.
    The situation is depicted in the right-hand side \cref{fig:families_matching_ladder}.
    Now the formula 
    \[
      \phi'(b,r) 
      :=
      \phi(b,r)
      \wedge 
      \neg \exists b'.
      \big (
      \mathrm{Blue}(b')
      \wedge b' \prec b
      \wedge \phi(b',r)
      \big )
    \]
    stating that $b$ is smallest (with respect to $\prec$) blue element $\phi$-connected to $r$, interprets a matching between~$S'$ and every $N_i'$ for $i \in [\ell]$. 
    Therefore, in the ladder case we arrive at the same contradiction as in the matching case using the formula $\phi'$ instead of $\phi$.
    As we introduced two additional color predicates and quantifiers, we obtain a formula with quantifier rank at most $q+1+2$ and ladder index $\ell$, which still leads to a contradiction.
       
    \paragraph*{Bounding the length of $\boldsymbol I$.}
    Since $S$ has size at most $s$, $I$ is obtained from $J$ by iteratively taking 
    $\Delta^{\Phi_s}_k$-indiscernible sequences (and removing up to three elements) for at most $s$ times.
    If $J$ has length $N(m) := g^k(m)$, where $g(m):= f(m) + 3$ and $f$ is the function from \Cref{thm:indiscernible_ramsey} with $\Delta=\Delta^{\Phi_s}_k$,
    then by said theorem, $I$ has length at least $m$, as required.

    \paragraph*{The special case for $\boldsymbol{a \in A \in \mathcal N_\alpha (I)}$.}
    Let $a$ be an element contained in $A_i = N_\alpha(\bar a_i) \in \mathcal N_\alpha (I)$ for some $i \in [m]$.
    We previously established that there exist $s_<,s_> \in S$ and a set $A_j = N_\alpha(\bar a_j) \in \mathcal N_\alpha (I)$ for some $j \in [m]$
    such that $a$ is $\phi$-equivalent to $s_<$ over all sets in $\mathcal N_\alpha(I)$ before $A_j$ and $\phi$-equivalent to $s_>$ over all sets in $\mathcal N_\alpha(I)$ after $A_j$.
    Assume towards a contradiction that $i \neq j$.
    Then $a$ must be $\phi$-equivalent to $s_<$ or $s_>$ over $A_i$. 
    Assume, by symmetry, $a$ and $s_<$ are $\phi$-equivalent over~$A_i$. 
    Since $a$ is contained in $A_i$, also $s_<$ must be, contradicting the fact that $S \cap \bigcup \mathcal{N}_\alpha (I) = \varnothing$.

    \paragraph*{The stable case.}
    Assume $\CC$ is monadically stable.
    As established previously for
    every $a \in \struc A$ there exists an index $\mathrm{ex}(a) \in [m]$ and $s_\sless,s_\sgreater \in S$ such that $a$ is $\mathrm{eq}^\phi_\alpha(s_\sless,x,\bar y)$-connected to $\bar a_1,\ldots,\bar a_{\mathrm{ex}(a)-1}$ and $\mathrm{eq}^\phi_\alpha(s_\sgreater,x,\bar y)$-connected to $\bar a_{\mathrm{ex}(a)+1},\ldots,\bar a_m$.
    By symmetry, we can assume $\mathrm{ex}(a) \leq \frac{m}{2}$,
    meaning that $a$ is $\mathrm{eq}^\phi_\alpha(s_\sgreater,x,\bar y)$-connected to at least two elements of $I$.
    Without loss of generality, we can assume $k$ to be large enough that
    by $\Delta^{\Phi_s}_k$-indiscernibility of $I$ and by \cref{thm:iseq_n_mon_stable}, $a$ must be $\mathrm{eq}^\phi_\alpha(s_\sgreater,x,\bar y)$-connected to all but at most one element of $I$.
    We can therefore assume that $\mathrm{ex}(a)$ points to this element and that $s_\sless = s_\sgreater$.
    \end{proof}

\subsection{Algorithmic aspects}\label{sec:algorithmic-aspects}

We now consider algorithmic aspects in the proof of \cref{thm:disjoint_families_nip}. In 
this section, $\CC$ is assumed to be a class of finite structures.
We show that certain sets of formulas admit indiscernible extractions running in fpt time, which will be required in \cref{sec:wideness}. 
First, we need the following notion of indiscernible extraction.

\begin{definition}\label{def:admitting_indiscernible_extraction}
Let $f:\NN \rightarrow \NN$ and $T:\NN \times \NN \rightarrow \NN$ be functions and $\Delta$ be a set of formulas.
We say that a  class $\CC$ of finite structures \emph{admits $(f,T,\Delta)$-indiscernible extraction}, if there exists an algorithm that, given a sequence $J$ of length $f(m)$ in any structure $\struc A \in \CC$, extracts a $\Delta$-indiscernible subsequence $I \subseteq J$ of length~$m$ in time $T(|J|,|\struc{A}|)$.
\end{definition}

\begin{theorem}\label{thm:disjoint_families_nip_algorithmic}
  Let $\CC$ be a class of finite structures
  such that  $\CC$ expanded with $k$ new constant symbols $\{c_1,\ldots,c_k\}$ admits $(f,T_\Phi,\Delta_k^\Phi)$-indiscernible extraction for 
$
\Phi(x,\bar y) 
:= 
\{\mathrm{eq}_\alpha^\phi(c_i,x,\bar y) ~:~ i \in [k]\}
$.
Then
in \cref{thm:disjoint_families_nip} 
we can set $N(n) := g^k(n)$, where $g(n) := f(n) + 3$, and there exists an algorithm that, given  $\struc A\in \CC$ and a sequence $J$ together with its $\alpha$-neighborhoods,
computes $I$ and $S$ in time 
$
  \mathcal{O}
  \big(
  k \cdot T_\Phi(|J|,|\struc A|)
  +
  k^2 \cdot F_\phi(|\struc A|) \cdot |\struc A|^2
  \big),
$
where $F_\phi$ is the time required to check whether $\struc A \models \phi(a,b)$ for two elements $a,b \in \struc A$.
\end{theorem}
\begin{proof}
We resume the proof of \cref{thm:disjoint_families_nip} and 
assume that the extension of $\CC$ with $s$ constant symbols admits $(f,T_\Phi,\Delta_k^{\Phi_k})$-indiscernible extraction. 
    Since $k \ge s$, in particular it admits $(f,T_\Phi,\Delta_k^{\Phi_s})$-indiscernible extraction.
    Let us now analyze the runtime required to construct $I$ and $S$.
    We assume the $\alpha$-neighborhoods of every element in $J$ are given as part of the input.
    Our extraction runs for at most $s$ rounds extracting a $\Delta_k^{\Phi_s}$-indiscernible sequence in time $T_\Phi(|J|,|\struc A|)$ in every round.
    Each round, we require additional time to select an element that is marked as the next constant.
    We require that it is $\phi$-inequivalent to all the previously marked constants over all but at most two of the $\alpha$-neighborhoods of the remaining subsequence of $J$.
    In order to find it, we iterate through every unmarked element of $a\in \struc A$, every $\alpha$-neighborhood $A$ of the remaining subsequence of $J$ and every constant $a'$ marked so far and check
    \begin{enumerate}
      \item whether $a$ and $a'$ are both contained or both not contained in $A$, and      
      \item whether $a$ and $a'$ have the same $\phi$-neighborhood in $A$.
    \end{enumerate}
    
    An efficient way to do this is to mark the $\phi$-neighborhood of every newly added constant.
    Then, since the $\alpha$-neighborhoods of $J$ are disjoint, comparing a single element to all constants takes at most time $\mathcal{O}(s \cdot F_\phi(|\struc A|) \cdot |\struc A|)$ and doing so for all elements adds $|\struc A|$ as a factor.
    Combining the running times we get an upper bound for the total running time of
    \[
    \mathcal{O}
    \big(
    \underbrace{s \cdot T_\Phi(|J|,|\struc A|)}
    _{\mathclap{\text{extracting indiscernibles}}}
    \quad
    +
    \quad
    \underbrace{
      s^2 \cdot F_\phi(|\struc A|) \cdot |\struc A|^2
    }
    _\text{selecting constants}
    \big).\qedhere
    \]
\end{proof}

We conclude this section by showing that certain sets of formulas admit indiscernible extractions running in fpt time.
To this end, for $r\in\NN$ let $\mathrm{dist}_{\leq r}(x,y)$ denote the formula in the signature of graphs expressing that the vertices $x$ and $y$ are at distance at most $r$ from each other.

\begin{lemma}\label{lem:extraction_for_wideness}
For every monadically stable class of graphs $\CC$ and $r,k\in\NN$ we have the following.
Let~$\CC^+$ be the class $\CC$ expanded with constant symbols $\{c_1,\ldots,c_k\}$ and define 
\[
\Phi(x,y) 
:= 
\{\mathrm{eq}_\alpha^\phi(c_i,x, y) ~:~ i \in [k]\},
\]
where $\phi(x,y) := E(x,y)$ and $\alpha(x,y) := \mathrm{dist}_{\leq r}(x,y)$.
There exists $t \in \NN$ such that $\CC^+$ admits $(f,T,\Delta_k^\Phi)$-indiscernible extraction for
\[
    f(m) := m^t \quad \text{and} \quad T(m,n) := g(k) \cdot n ^3,
\]
where $g$ is a function independent of $\CC$, $r$ and $k$.
\end{lemma}
\begin{proof}
As we are aiming for fpt running times, in the following we will denote by $g_1,g_2,g_3$ functions whose existence will be clear from the context.

As $\CC^+$ is a monadic expansion of $\Cc$, also $\Cc^+$ is monadically stable.
Note that the size of~$\Delta^\Phi_k$ only depends on $k$, and formulas in $\Delta^\Phi_k$ contain at most $k$ free variables.
Therefore, by \cref{thm:poly_seqs}, there exists a $t\in \NN$ such that in any $n$-vertex graph $G\in\CC^+$ we can extract $\Delta^\Phi_k$-indiscernible sequences of length $m$ from sequences of length $m^t$ in time $\Oof(g_1(k) \cdot n^{1/k} \cdot T_\star)$, where $T_\star$ is the time it takes to check whether $G \models \gamma(a_1,\ldots,a_k)$ for a formula $\gamma\in \Delta^\Phi_k$ and given elements $a_1,\ldots,a_k \in G$.
Let us now bound~$T_\star$.
By construction, we know that every formula in $\Delta^\Phi_k$ is of the form
\[
    \gamma(y_1,\ldots,y_{i}) = \exists x.\bigwedge_{j\in[i]}\phi_j(x,y_j)
\]
for some $i \leq k$, where each $\phi_j$ is a boolean combination of formulas from $\Phi$.
A single formula in $\Phi$ is of the form
\[
    \mathrm{eq}_\alpha^\phi(c,x,y) 
    =
    \big(
    \mathrm{dist}_{\leq r}(x, y) 
    \leftrightarrow    
    \mathrm{dist}_{\leq r}(c, y) 
    \big)
    \wedge
    \forall w.\, \mathrm{dist}_{\leq r}(y, w) 
    \rightarrow 
    \big (E(x,w) \leftrightarrow E(c,w) \big ),
\]
i.e., it asks whether $x$ and a constant element $c$ that is marked in $G$ have the same edge-neighborhood in the distance-$r$ ball around $y$ and are both contained or both not contained in it.
We assume that as a preprocessing step at the beginning of the extraction, we have calculated the distance-$r$ ball of every element in $G$ and therein marked the edge-neighborhood of every constant $\{c_1,\ldots,c_k\}$.
This can be done in time $\Oof(k \cdot n^3)$ by performing depth first searches.
After this preprocessing we can evaluate 
$G \models \mathrm{eq}_\alpha^\phi(c,a_1,a_2)$ in time $\Oof(n)$ for any $a_1,a_2 \in G$ and $c\in\{c_1,\ldots,c_k\}$.
Adding a linear factor for the existential quantifier, we can evaluate $\gamma$ in time $\Oof(g_2(k)\cdot n^2)$, which bounds $T_\star$.
Plugging $T_\star$ into the bound above given by \cref{thm:poly_seqs} and adding the cubic time preprocessing we get an overall running time of $\Oof(g_3(k) \cdot n^3)$ as desired.
\end{proof}

\section{Flatness in monadically stable classes of graphs}\label{sec:wideness}

In this section we use \cref{thm:disjoint_families_nip} to characterize monadically stable graph classes in terms of flip-flatness. 
We restate the definitions and start with the forward direction.
A \emph{flip} in a graph $G$ is a pair of sets $\mathsf{F} = (A,B)$ with $A,B \subseteq V(G)$. We write $G \oplus \mathsf{F}$ for the graph $(V(G), \{uv ~:~ uv \in E(G) \text{ xor } (u,v) \in (A \times B) \cup (B \times A)\})$.
For a set $F =\{\flip F_1,\ldots,\flip F_n\}$ of pairwise distinct flips we write $G\oplus F$ for the graph $G \oplus \flip F_1 \oplus \dots \oplus \flip F_n$.
Given a graph $G$ and a set of vertices $A\subseteq V(G)$ we call $A$ \emph{distance-$r$ independent} if all vertices in $A$ have pairwise distance greater than $r$ in $G$.

\defflipwide


\thmfuqwforward*
\begin{proof}
    Let $\CC$ be a monadically stable class of graphs and $r\in \NN$. 
    We want to show that in every graph $G\in \CC$, in every set $A \subseteq V(G)$ of size $N_{r}(m)$ we find a subset $B\subseteq A$ of size at least $m$ and a set $F$ of at most $s_{r}$ flips such that $B$ is distance-$r$ independent in $G\oplus F$.
    We will first inductively describe how to obtain the set of vertices $B$ and the set of flips $F$ and bound the runtime and values for $N_r$ and $s_r$ later.
    In the base case we have $r=0$. We can pick $B := A$ and $F := \varnothing$ and there is nothing to show.

    In the inductive case let $r = 2i + p$ for some $i\in \NN$ and parity $p\in \{0,1\}$.
    Apply the induction hypothesis to obtain $s_r$ flips $F_r$ and a set $A_r$ that  is distance-$r$ independent in $G_r := G \oplus F_r$.
    Note that for every fixed number $t$ of flips, since flips are definable by coloring and a quantifier-free formula, the class $\Cc_t$ of all graphs obtainable from graphs of $\Cc$ by at most $t$ flips is monadically stable by~\cref{lem:definable-expansions}. Hence, $G_r$ comes from the monadically stable class~$\CC_{s_r}$.
    Our goal is to find a set of flips $F'_{r+1}$ (of fixed finite size which will determine the number $s_{r+1}$) together with a set $A_{r+1} \subseteq A_r$, that is distance-$(r+1)$ independent in $G_{r+1} := G_r \oplus F'_{r+1}$.
    Then $G_r = G \oplus F_{r+1}$ with $F_{r+1} = F_r \cup F'_{r+1}$.

    Since the elements in $A_r$ have pairwise disjoint distance-$i$ neighborhoods, we can apply \cref{thm:disjoint_families_nip} to $A_r$ with $\phi(x,y)=E(x,y)$ and $\alpha(x,y) = \mathrm{dist}_{\leq i}(x,y)$.
    Since we are in the monadically stable case, this yields a subset $A_{r+1} \subseteq A_r$ and a small set of sample vertices $S$ such that for every vertex $a\in V(G)$ there exists $s(a) \in S$ and $\mathrm{ex}(a)\in A_{r+1}$ such that $a$ has the same edge-neighborhood as $s(a)$ in the distance-$i$ neighborhoods of all elements from $A_{r+1}\setminus \{\mathrm{ex}(a)\}$.
    We now do a case distinction depending on whether $r$ is even or odd.

    \paragraph*{The even case: $\boldsymbol{r=2i}$.}
    Our goal is to flip away every edge between a pair of two elements contained in distinct distance-$i$ neighborhoods of elements from $A_{r+1}$.
    Let us first show that
    there exists a coloring $\mathrm{col} : V(G) \rightarrow \KK$, where $\KK$ is a set of at most $|S|\cdot 2^{|S|}$ many colors and a relation $R\subseteq \KK^2$, 
    such that for any two elements $a_1$, $a_2$ contained in distinct distance-$i$ neighborhoods of elements from $A_{r+1}$,
    \begin{equation}\label{eq:color_iff_edge}
    G_r \models E(a_1,a_2)
    \hspace{2mm} \iff \hspace{2mm}
    (\mathrm{col}(a_1),\mathrm{col}(a_2)) \in R.
    \end{equation}
    
    We construct the coloring as follows.
    To every $a \in V(G)$ we assign a color $\mathrm{col}(a)$ from $S\times \mathcal{P}(S)$, where
    in the first coordinate we store $s(a)$ and in the second coordinate 
    we store the set of all $b\in S$ with $G_r \models E(a,b)$. We define $R$ to contain all pairs $((s(x),X), (s(y),Y))$, where $x,y\in V(G_r)$ and $X,Y\subseteq S$, such that $s(x)\in Y$. 
    To argue the correctness of this coloring,
    assume now we have a vertex $a_1$ in the distance-$i$ neighborhood of a vertex $a_1'\in A_{r+1}$ and a vertex $a_2$ in the distance-$i$ neighborhood of a vertex $a_2'\in A_{r+1}$.
    Let $a_1' \neq a_2'$
    and $(s(a_1),C_1)$ and $(s(a_2),C_2)$ be the colors of $a_1$ and $a_2$ respectively.
    By \cref{thm:disjoint_families_nip} we can assume that $\mathrm{ex}(a_1) = a_1'$ and $\mathrm{ex}(a_2) = a_2'$.
    Note that~$a_2$ is not contained in the distance-$i$ neighborhood of $\mathrm{ex}(a_1)$, since otherwise $\mathrm{ex}(a_1) = \mathrm{ex}(a_2)$.
    Thus, $a_1$ has the same adjacency to $a_2$ as $s(a_1)$, yielding
    \[
        G_r \models E(a_1,a_2) 
        \iff 
        G_r \models E(s(a_1),a_2)
         \iff
        s(a_1) \in C_2
        \iff
        \big((s(a_1),C_1),(s(a_2),C_2)\big) \in R.
    \]

    Let $D$ be all vertices with distance exactly $i$ in $G_{r}$ to one of the elements in $A_{r+1}$.
    We may assume our colors to be ordered by some total order $\le$.
    Based on the coloring, we construct the set of flips $F'_{r+1}$ containing for every $C_1 \le C_2 \in \mathcal K$ with $(C_1,C_2) \in R$ the flip 
    $ (\{ x \in D ~:~ \mathrm{col}(x) = C_1 \}, \{ x \in D ~:~ \mathrm{col}(x) = C_2 \})$.
    
    Let us now argue that $A_{r+1}$ is distance-$(r+1)$ independent in $G_{r+1}$.
    Since we only flipped edges between vertices at distance exactly $i$ to vertices from $A_{r+1}$, we have that $A_{r+1}$ remains distance-$2i$ independent in $G_{r+1}$.
    Assume therefore towards a contradiction that there exists a path $P = (a,a_1,\ldots,a_i,b_i,\ldots,b_1,b)$ of length $r+1=2i+1$ between two distinct vertices $a,b \in A_{r+1}$ in $G_{r+1}$.

    By symmetry, we can assume $\mathrm{col}(a_i) \le \mathrm{col}(b_i)$.
    Note that there exists either zero or one flip in $F'_{r+1}$ flipping the adjacency between $a_i$ and $b_i$,
    with one flip existing between the two vertices if and only if $(\mathrm{col}(a_i),\mathrm{col}(b_i)) \in R$.
    By (\ref{eq:color_iff_edge}), this is equivalent to $a_i$ and $b_i$ being adjacent in $G_r$.
    Thus, going from $G_r$ to $G_{r+1}$, the adjacency between $a_i$ and $b_i$ is flipped if and only if they were adjacent in $G_r$,
    meaning that $a_i$ and $b_i$ are non-adjacent in $G_{r+1}$.
    We can conclude that $P$ is not a path in~$G_{r+1}$, and that $A_{r+1}$ is distance-$(r+1)$ independent in $G_{r+1}$.

    \paragraph*{The odd case: $\boldsymbol{r=2i + 1}$.}
    For every $s \in S$ let $C_s$ be the set containing every vertex $a$ for which we have $s(a) = s$ and that is at distance at least $i+1$ from every vertex in $A_{r+1}$.
    Let~$D_s$ be the set containing every edge-neighbor of $s$ that is at distance exactly $i$ from one of the vertices in $A_{r+1}$.
    We now build the set of flips $F'_{r+1}$ by adding for every sample vertex $s \in S$ the flip $(C_s,D_s)$.

    Let us now argue that $A_{r+1}$ is distance-$(r+1)$ independent in $G_{r+1}:= G_r$.
    We only flip edges in $G_r$ between pairs of vertices $a,b$
    such that $a$ has distance (in $G_r$) at least $i+1$ and $b$ has distance exactly $i$ to $A_{r+1}$.
    It follows that $A_{r+1}$ remains distance-$(2i + 1)$ independent in $G_{r+1}$.
    Assume towards a contradiction that there exists a path $P = (a,a_1,\ldots,a_i,u,b_i,\ldots,b_1,b)$ of length $r+1=2i+2$ between two vertices $a,b \in A_{r+1}$ in~$G_{r+1}$.
    The distance between $a$ and $a_i$ (resp. $b$ and $b_i$) is not affected by the flips. Only the connection between $a_i$ (resp. $b_i$) and $u$ can possibly be impacted.
    Additionally, note that $u\in C_{s(u)}$.

    Since $a$ and $b$ are distinct we have that either $\mathrm{ex}(u) \neq a$ or $\mathrm{ex}(u) \neq b$. By symmetry, we can assume the former case.
    As $a_i$ is in the distance-$i$ neighborhood of $a$, we have $G_r\models E(u,a_i) \iff  G_r \models E(s(u),a_i)$.
    Observe that if $E(s(u),a_i)$ holds in $G_r$, then $a_i \in D_{s(u)}$ and therefore the edge $(u,a_i)$ was removed by the flips and is not in $G_{r+1}$.
    Similarly, if $E(s(u),a_i)$ does not hold in $G_r$, then $a_i \not\in D_{s(u)}$
    and the edge $(u,a_i)$ was not introduced by any flip.
    We can conclude that $P$ is not a path in~$G_{r+1}$, and that
    $A_{r+1}$ is distance-$(r+1)$ independent in $G_{r+1}$.

    \paragraph*{Size bounds and runtime analysis.}
    We will analyze the bounds $N_r$ and $s_r$ by following the inductive construction of $A_r$ and $F_r$.
    In the base case we have $N_{0}(m) := m$ and $s_{0} := 0$.

    Before handling the inductive step, we will first analyze runtime and size bounds of \cref{thm:disjoint_families_nip}, which is the main tool used during the construction.
    For $i \in \NN$ we denote by $\CC_i$ the monadically stable class of all graphs obtainable by performing at most $s_i$ flips on graphs of $\CC$.
    Let~$k_{i}$ be the sample set bound given by \cref{thm:disjoint_families_nip} for the class $\CC_{i}$ and the formulas $\phi(x,y) := E(x,y)$ and $\alpha_i(x,y) := \mathrm{dist}_{\leq i}(x,y)$.
    Observe that $k_{i}$ depends only on $\CC$ and $i$.
    Let $\Phi_{i}(x,y)$ be the corresponding set of formulas specified in \cref{thm:disjoint_families_nip} for $\phi$ and $\alpha_i$.
    Remember the notion of \emph{indiscernible extraction} given by \Cref{def:admitting_indiscernible_extraction}.
    Since $\phi$ is the edge relation and $\alpha_i$ is a distance formula we can apply \cref{lem:extraction_for_wideness} and we get that $\CC_i$ expanded with $k_{i}$ constant symbols admits 
    \[
        \big(m \mapsto m^{t_i}, (m,n) \mapsto g(k_{i}) \cdot n^3, \Delta^{\Phi_{i}}_{k_{i}} \big)
        \text{-indiscernible extraction}
    \]
    for an integer $t_i$ depending again only on $\CC$ and $i$ and a function $g$. 
    Plugging this result back into \cref{thm:disjoint_families_nip_algorithmic}
    we see that in the class $\CC_i$ for $\phi$ and $\alpha_i$ we can extract the desired subsequences of length~$m$ from sequences of length $m^{t_i'}$
    in time $\Oof(h_i(\CC,i) \cdot n^3)$ for an integer $t_i'$ and function $h_i$.

    In the inductive case we first apply the induction hypothesis resulting in a graph from $\CC_{i}$.
    We work on the set $A_i$ of size $N_i(m)$ where we apply \cref{thm:disjoint_families_nip} to extract $A_{i+1}$. 
    We can therefore set $N_{i+1}(m) := N_i(m)^{t_i'} = m^{t_1' \cdot \ldots \cdot t_i'}$.
    As desired, $N_{i+1}$ depends only on $\CC$ and $i$.

    To count the flips in the inductive cases, notice that
    in the odd case we create a flip for each of the~$k_i$ sample vertices.
    In the even case we create a flip for each pair of colors 
    with the number of colors being bounded by $k_i\cdot 2^{k_i}$.
    We can therefore set 
    $s_{i+1} := s_i + (k_i\cdot 2^{k_i})^2 = \sum_{j \in [i]}  (k_j\cdot 2^{k_j})^2$.
    As desired, $s_{i+1}$ depends only on $\CC$ and $i$.

    Regarding the running time, we see that the time needed to compute $A_{i+1}$ and $F_{i+1}$ is dominated by applying \cref{thm:disjoint_families_nip} $i$ times.
    We can therefore bound the total run time of our construction by
    $\Oof(f_\CC(r) \cdot n^3)$ for some function~$f_\CC$.

\end{proof}

We have shown that for graph classes, monadic stability implies flip-flatness. We now show that the reverse holds as well.
We will use the following statement which is an immediate consequence of Gaifman's locality theorem~\cite{gaifman82}.
For an introduction of the locality theorem see for example~\cite[Section 4.1]{grohe2008logic}.

\begin{restatable}[of {\cite[Main Theorem]{gaifman82}}]{corollary}{gaifmancoloring}\label{lem:gaifman-coloring}
      Let $\phi(x,y)$ be a formula. Then there are numbers $r,t\in \NN$, where $r$ depends only on the quantifier-rank of $\phi$ and $t$ depends only on the signature and quantifier-rank of $\phi$,
      such that every (colored) graph $G$ can be vertex-colored using $t$ colors
      in such a way that for any two vertices $u,v\in V(G)$ with distance greater than $r$ in $G$,
      $G\models\phi(u,v)$ depends only on the colors of $u$ and $v$.
      We call $r$ the \emph{Gaifman radius} of $\phi$.
\end{restatable}

\begin{lemma}\label{lem:fuqw_backward}
   Every flip-flat class of graphs is monadically stable.
\end{lemma}

\begin{proof}
    Assume towards a contradiction that there exists a class $\CC$ that is not monadically stable but flip-flat.
By definition of monadic stability there exists a formula $\sigma(x,y)$ defining arbitrarily large orders in a coloring of~$\CC$, that is, for every $n\in \NN$ there exists a graph $G \in \CC$ and a coloring $G^+$ such that we find a sequence $(a_1,\ldots,a_n)$ in $G^+$ with  $G^+ \models \sigma(a_i,a_j)$ if and only if $i<j$ for all $i,j \in [n]$.
    
    Let $r$ be the Gaifman radius of $\sigma$ as given by \cref{lem:gaifman-coloring} (depending only on the quantifier-rank of~$\sigma$).
    Let $N_r$ and $s_r$ be the size function and number of flips we obtain from $\Cc$ being flip-flat with radius $r$. As stated in \cref{lem:gaifman-coloring}, let $t_r$ be the number of colors needed to determine the truth value of formulas in
    the signature of graphs with $2s_r$ additional unary predicates and the same quantifier-rank as $\sigma(x,y)$. Let $n:=N_r(t_r + 1)$ and fix a graph $G\in \Cc$ such that in $G^+$ we find a sequence $I$ of length $n$ ordered by $\sigma$.
    
We rewrite $\sigma$ into a formula $\sigma_{r}$ with the same quantifier rank as~$\sigma$ such that for every set $F$ of at most~$s_r$ flips there exists a coloring $H^+$ of $H := G \oplus F$ such that $G^+ \models \sigma(u,v)$ if and only if $H^+ \models \sigma_{r}(u,v)$ for all $u,v \in V(G)$. This translation depends only of $\sigma$ and $s_r$ and its existence is easily proven by induction on the number of flips by the following translation of atomic formulas $E(x,y)$. If a flip $\mathsf F=(A,B)$ is marked by two unary predicates~$A$ and $B$, then for all vertices $u,v$ we have $G\models E(u,v)$ if and only if 
$G \oplus \mathsf F\models E(u,v)~\mathrm{xor}~((u\in A\wedge v\in B) \vee (u\in B\wedge v\in A))$. 
        
    We apply flip-flatness to $I$ and find a subsequence $J \subseteq I$ of length $t_r+1$ together with a set~$F$ of at most $s_r$ many flips such that $J$ is $r$-independent in $H = G \oplus F$. By construction, $\sigma_r$ orders $J$ in $H^+$. 
    
As $\sigma_r$ has the same quantifier-rank as $\sigma$ and is a formula over the signature of graphs extended by~$s_r$ unary predicates, by \cref{lem:gaifman-coloring} there exists a coloring 
of $H^+$ with $t_r$ colors such that the truth of $\sigma_r(u,v)$ only depends on the colors of $u$ and $v$ for all $u,v \in J$.
By the pigeonhole principle there exist two distinct vertices $u,v \in J$ that are assigned the same color. 
We therefore have $H^+ \models \sigma_r(u,v) \iff \sigma_r(v,u)$,
which is a contradiction to $\sigma_r$ ordering $J$ in~$H^+$.
\end{proof}

From \cref{thm:fuqw_forward} and \cref{lem:fuqw_backward} we conclude the following.

\thmfuqw*

\bibliography{ref}

\end{document}